\documentclass[preprint,3p,11pt]{elsarticle}
\usepackage{amsmath,amsthm,amsfonts,amssymb,mathrsfs}
\def\R{{\mathbb R}}
\def\C{{\mathbb C}}
\def\mS{{\mathcal S}}
\DeclareMathOperator{\Tr}{Tr}
\DeclareMathOperator{\spec}{Spec}
\newtheorem{theorem}{Theorem}[section]
\newtheorem{proposition}[theorem]{Proposition}

\newtheorem{remark}{Remark}[section]

\usepackage{graphicx}
\usepackage{subfigure}

\usepackage[latin1]{inputenc}
\usepackage[swedish,english]{babel}

\journal{Applied and Computational Harmonic Analysis}

\begin{document}

\begin{frontmatter}
\title{On the polarizability and capacitance of the cube}
\author{Johan Helsing\corref{cor1}}
\cortext[cor1]{Corresponding author, Tel.:+46 46 2223372.}
\ead{helsing@maths.lth.se}
\author{Karl-Mikael Perfekt}
\ead{perfekt@maths.lth.se}
\address{Centre for Mathematical Sciences\\
  Lund University, Box 118, SE-221 00 Lund, Sweden}
\begin{abstract}
  An efficient integral equation based solver is constructed for the
  electrostatic problem on domains with cuboidal inclusions. It can be
  used to compute the polarizability of a dielectric cube in a
  dielectric background medium at virtually every permittivity ratio
  for which it exists. For example, polarizabilities accurate to
  between five and ten digits are obtained (as complex limits) for
  negative permittivity ratios in minutes on a standard workstation.
  In passing, the capacitance of the unit cube is determined with
  unprecedented accuracy. With full rigor, we develop a natural
  mathematical framework suited for the study of the polarizability of
  Lipschitz domains. Several aspects of polarizabilities and their
  representing measures are clarified, including limiting behavior
  both when approaching the support of the measure and when deforming
  smooth domains into a non-smooth domain. The success of the
  mathematical theory is achieved through symmetrization arguments for
  layer potentials.
\end{abstract}
\end{frontmatter}

\section{Introduction}

The determination of polarizabilities and capacitances of inclusions
of various shapes has a long history in computational
electromagnetics. Inclusions with smooth surfaces are, by now, rather
standard to treat. When surfaces are non-smooth, however, the
situation is different. Numerical solvers can run into problems
related to stability and resolution. Particularly so in three
dimensions and for certain permittivity combinations. Solutions may
not converge or results could be hard to interpret.
See~\cite{Perr10,Pitko10,Walle08} and references therein. The
situation on the theoretical side is similar. When, and in what sense
do solutions exist? Such questions are in the mainstream of
contemporary research in harmonic analysis. Coincidentally, also in
applied physics (plasmonics) there is a growing interest in solving
electrostatic problems on domains with structural discontinuities and
a concern about the sufficiency of available
solvers~\cite{Grill11,Zhang11}.

This paper addresses several fundamental issues related to the
problems just mentioned. We construct a stable solver for the
polarizability and capacitance of a cube based on an integral equation
using the adjoint of the double layer potential. We compute solutions
of unprecedented accuracy and interpret the results within a rigorous
mathematical framework. The reason for working with a cube are
twofold. First, the cube has the advantage that its geometric
difficulties are concentrated to edges and corners, since its faces
are flat. Integral equation techniques, which often excel for boundary
value problems in two dimensions, typically suffer from loss of
accuracy in the discretization of weakly singular integral operators
on curved surfaces in three dimensions. Here we need not worry about
that. Secondly, cubes are actually common in plasmonic applications.

In the purely theoretical sections we begin by collecting a number of
results and recent advances in the theory of layer potentials
associated with the Laplacian in Lipschitz domains. The most obvious
reason for this is that the invertibility study of layer potentials
leads to the solution of the boundary value problem implicit in the
definition of polarizability, and is as such the basis for both the
mathematical and numerical aspects of this paper. Furthermore, the
properties of the polarizability for a non-smooth domain such as a
cube are quite subtle, and it is our ambition to provide a solid
theoretical foundation for the problem at hand, giving a careful and
detailed exposition of a mathematical framework that clarifies a
number of points.

Since the double layer potential is not self-adjoint in the
$L^2$-pairing, we develop certain symmetrization techniques for it, in
particular extending the work of Khavinson, Putinar and Shapiro
\cite{Khav91} to the case of a non-smooth domain. These techniques are
used to prove the unique existence of the polarizability itself for a
Lipschitz domain, as well as of a corresponding representing measure
\cite{Gold83}. We present a thorough discussion of the smooth case,
the limiting behavior in passing from the smooth to the non-smooth
case, and ultimately the general case. Concerning the last point, a
condition ensuring that the representing measure has no singular part
is given, and it is proven that in the support of the absolutely
continuous part of the measure, the polarizability can not be given a
direct interpretation in terms of a potential with finite energy
solving the related boundary value problem.

The paper is organized as follows: Section~\ref{sec:elstat} formulates
the electrostatic problem and defines the polarizability. Existence
issues and representations are reviewed in Section~\ref{sec:theory}.
For ease of reading, rigorous statements and proofs are deferred to
Sections~\ref{sec:prelim} and~\ref{sec:results}. The capacitance is
discussed in Section~\ref{sec:capa}. Section~\ref{sec:strat} reviews
the state of the art with regard to numerical schemes.
Section~\ref{sec:two} gives a necessary background to the present
solver. New development takes place in Section~\ref{sec:three}. The
last sections contain numerical examples performed in {\sc Matlab}.
Section~\ref{sec:numsquare} illustrates the effects of rounding
corners and Section~\ref{sec:numcube} is about the cube.

The main conclusion of the paper is that, from a numerical viewpoint,
it is an advantage to let cubes have sharp edges and corners as
opposed to the common practice of rounding them slightly. Furthermore,
the representing measure for the polarizability of the cube is
determined, and a new benchmark for the capacitance of the unit cube
is established.

\section{The electrostatic problem and the polarizability}
\label{sec:elstat}

Let a domain $V$, an inclusion with surface $S$ and permittivity
$\epsilon_2$, be embedded in an infinite space. The exterior to the
closure of $V$ is denoted $E$ and has permittivity $\epsilon_1$. Let
$\nu_r$ be the exterior unit normal of $S$ at position $r$.

We seek a potential $U(r)$, continuous in $E\cup S\cup V$, which
satisfies the electrostatic equation
\begin{equation}
\Delta U(r)=0\,, \qquad r\in E\cup V\,,
\label{eq:elstat}
\end{equation}
subject to the boundary conditions on the limits of normal derivatives
\begin{equation}
\epsilon_1\frac{\partial}{\partial\nu_r}U^{\rm ext}(r)=
\epsilon_2\frac{\partial}{\partial\nu_r}U^{\rm int}(r)
\label{eq:elstatb1}
\end{equation}
and behavior at infinity
\begin{equation}
\lim_{r\to\infty}\nabla U(r)=e\,.
\label{eq:elstatb2}
\end{equation}
Here superscripts {\footnotesize ext} and {\footnotesize int} denote
limits from the exterior or interior of $S$, respectively, and $e$ is
an applied unit field. Eqs.~(\ref{eq:elstat}), (\ref{eq:elstatb1}),
(\ref{eq:elstatb2}) constitute a partial differential equation
formulation of the electrostatic problem.
Proposition~\ref{prop:solsense} gives a strict interpretation of what
it means for a potential $U(r)$ to solve this problem, in particular
expressing~(\ref{eq:elstatb1}) in a distribution sense.

For the construction of solutions to (\ref{eq:elstat}),
(\ref{eq:elstatb1}), (\ref{eq:elstatb2}) we make use of fundamental
solutions to the Laplace equation in two and three dimensions
\begin{equation}
G(r,r')=-\frac{1}{2\pi}\log|r-r'| \qquad {\rm and} \qquad
G(r,r')=\frac{1}{4\pi}\frac{1}{|r-r'|}\,,
\label{eq:newtker}
\end{equation}
and represent $U(r)$ in terms of a single layer density $\rho(r)$ as
\begin{equation}
U(r)=e\cdot r+\int_S G(r,r')\rho(r')\,{\rm d}\sigma_{r'}\,,
\label{eq:rep1}
\end{equation}
where ${\rm d}\sigma$ is an element of surface area.

The representation~(\ref{eq:rep1}) satisfies~(\ref{eq:elstat})
and~(\ref{eq:elstatb2}). Its insertion in~(\ref{eq:elstatb1}) gives
the integral equation for $\rho(r)$
\begin{equation}
\rho(r)+2\lambda\int_S\frac{\partial}{\partial \nu_r}G(r,r')\rho(r')
\,{\rm d}\sigma_{r'}=-2\lambda\left(e\cdot \nu_r\right)\,,
\quad r\in S\,,
\label{eq:inteq1}
\end{equation}
where the parameter
\begin{equation}
\lambda=\frac{\epsilon_2-\epsilon_1}{\epsilon_2+\epsilon_1}\,.
\end{equation}

The polarizability tensor of $V$ can be defined in terms of an
integral over a polarization field. When $V$ features sufficient
symmetry, such as the octahedral symmetry of the cube, the
polarizability is isotropic and reduces to a scalar
$\alpha(\epsilon_1,\epsilon_2)$, see~\cite{Sihvo04}, which can be
determined via
\begin{equation}
\left(\epsilon_2-\epsilon_1\right)\int_V\nabla U(r)\,{\rm d}r=
{\alpha}(\epsilon_1,\epsilon_2)e\,,
\label{eq:pola1}
\end{equation}
where ${\rm d}r$ is a volume element. Using integration by parts
in~(\ref{eq:pola1}) it is possible to express
$\alpha(\epsilon_1,\epsilon_2)$ as an integral over $\rho(r)$
\begin{equation}
\alpha(\epsilon_1,\epsilon_2)=-\epsilon_1
\int_S\rho(r)\left(e\cdot r\right)\,{\rm d}\sigma_{r}\,.
\label{eq:alpha}
\end{equation}
More generally, the components of the polarizability tensor can be
recovered via integrals similar to that in~(\ref{eq:alpha}) using
different applied fields $e$, but in what follows we tacitly assume
that $V$ is sufficiently symmetric as to allow for a scalar valued
$\alpha(\epsilon_1,\epsilon_2)$.

\section{Theory -- overview}
\label{sec:theory}

For what surface shapes $S$ and permittivities $\epsilon_1$,
$\epsilon_2$ does the electrostatic problem have a solution? Starting
with~(\ref{eq:inteq1}) and partly following~\cite{Mayer05}, this
section sketches the derivation of some important existence results.
It also motivates an integral representation formula for
$\alpha(\epsilon_1,\epsilon_2)$ and two sum rules which are used for
validation in our numerical experiments.

\subsection{Existence of solutions for smooth $S$}

Let us rewrite~(\ref{eq:inteq1}) in the abbreviated form
\begin{equation}
\left(I+\lambda K\right)\rho(r)=\lambda g(r)\,,
\label{eq:inteq2}
\end{equation}
where $I$ is the identity. If $S$ is smooth, then~(\ref{eq:inteq2}) is
a Fredholm second kind integral equation with a compact, non-self
adjoint, integral operator $K$ whose spectrum is discrete and
accumulates at zero. Let $K$ and its adjoint $K^*$ (the double layer
potential) have eigenvectors $\phi_i$ and $\psi_i$ with corresponding
eigenvalues $z_i$. All eigenvalues are real and bounded by one in
modulus.  Non-zero eigenvalues have finite multiplicities. Normalizing
\begin{equation}
\int_S \overline{\psi_i(r)}\phi_j(r)\,{\rm d}\sigma_r=\delta_{ij}\,,
\end{equation}
the kernel of $K$ can be written
\begin{equation}
K(r,r')=\sum_i z_i\phi_i(r)\overline{\psi_i(r')}\,.
\label{eq:kernel}
\end{equation}
See, further, the discussion in Section~\ref{sec:prop}.

Let us introduce a new variable $z$ and a scaled polarizability
$\alpha(z)$ as
\begin{equation}
z=-1/\lambda\,,
\label{eq:obvious2}
\end{equation}
\begin{equation}
\alpha(z)\equiv\frac{\alpha(\epsilon_1,\epsilon_2)}{|V|\epsilon_1}\,,
\label{eq:alphaz}
\end{equation}
where $|V|$ is the volume of $V$. Then~(\ref{eq:alphaz}),
with~(\ref{eq:alpha}), can be written in the abbreviated form
\begin{equation}
\alpha(z)=\int_S h(r)\rho(r)\,{\rm d}\sigma_r\,.
\label{eq:alpha2}
\end{equation}
The relation~(\ref{eq:obvious2}) allows us to use the parameter
$\lambda$ or its negative reciprocal $z$, depending on what is most
convenient in a given situation.

In terms of the quantities
\begin{equation}
u_i=\int_S h(r)\phi_i(r)\,{\rm d}\sigma_r \quad {\rm and} \quad
v_i=\int_S \overline{\psi_i(r)}g(r)\,{\rm d}\sigma_r\,,
\label{eq:uivi}
\end{equation}
and using~(\ref{eq:kernel}) to construct the resolvent
of~(\ref{eq:inteq2}), one can write
\begin{equation}
\rho(r)=\sum_i\frac{\phi_i(r)v_i}{z_i-z}
\label{eq:lambda4a}
\end{equation}
and
\begin{equation}
\alpha(z)=\sum_i\frac{u_iv_i}{z_i-z}\,,
\label{eq:lambda4b}
\end{equation}
see Theorem~\ref{thm:smooth}. This suggests that neither $U(r)$ nor
$\alpha(z)$ exists for $z=z_i$ when $u_iv_i\ne 0$. There is an {\it
  electrostatic resonance} or {\it plasmon} at $z_i$.

For ease of interpretation, the sum in~(\ref{eq:lambda4b}) can be
considered as taken over distinct eigenvalues and with $u_iv_i$, for a
degenerate eigenvalue, being the sum of all residues belonging to that
eigenvalue. Then all $u_iv_i$ are non-negative and plasmons can be
classified as {\it bright} or {\it dark} depending on whether
$u_iv_i>0$ or not~\cite{Zhang11}. When $S$ is a circle, there are only
two eigenvalues: $z_1=-1$ which is simple and corresponds to a dark
plasmon and $z_2=0$ which has infinite multiplicity and corresponds to
a bright plasmon. When $S$ is a sphere, the eigenvalues are
$z_i=1/(1-2i)$. The multiplicity of $z_i$ is $2i-1$. The only bright
plasmon is associated with $z_2$.

For later reference we observe that the sum of all residues is
\begin{equation}
\sum_i u_iv_i=\int_S h(r)g(r)\,{\rm d}\sigma_r=2\,.
\label{eq:resid1}
\end{equation}
When the polarizability is isotropic one can, using techniques
from~\cite{Fuchs76}, also derive a weighted sum rule
\begin{equation}
\sum_i z_iu_iv_i=
\int_S\int_S h(r)K(r,r')g(r')\,{\rm d}\sigma_r'\,{\rm d}\sigma_r
=2(2/d-1)\,,
\label{eq:resid2}
\end{equation}
where $d=2,3$ is the dimension.

\subsection{Existence of solutions for non-smooth $S$}

If $S$ is gradually transformed from a smooth surface into a
non-smooth surface, eigenvalues $z_i$ travel and occupy a certain
subset of the interval $[-1,1]$ ever more densely. When $S$ ceases to
be smooth, $K$ is no longer compact with discrete eigenvalues. Rather,
$K$ has a continuous spectrum which on a certain function space
coincides with the aforementioned subset, accompanied by discrete
values. Disregarding the discrete spectrum, which for squares and
cubes turns out to correspond to dark plasmons, the
sum~(\ref{eq:lambda4b}) assumes a limit
\begin{equation}
\alpha(z)
=\int_\mathbb{R}\frac{{\rm d}\mu(x)}{x-z}
=\int_{\sigma_{\mu}}\frac{\mu'(x)\,{\rm d}x}{x-z}\,,
\label{eq:lambda5}
\end{equation}
where the measure $\mu(x)$ is real and non-negative and
$\sigma_{\mu}=\left\{x:\mu'(x)>0\right\}$. Here we have ignored the
possible presence of a singular spectrum. For further details and a
condition that serves to exclude this complication, see
Theorems~\ref{thm:muexist} and Section~\ref{sec:prop}. The sum
rules~(\ref{eq:resid1}) and~(\ref{eq:resid2}) assume the forms
\begin{align}
\int_\mathbb{R}\mu'(x)\,{\rm d}x&=2\,,
\label{eq:sum1}\\
\int_\mathbb{R}x\mu'(x)\,{\rm d}x&=2(2/d-1)\,.
\label{eq:sum2}
\end{align}

The numerical results in Sections~\ref{sec:numsquare}
and~\ref{sec:numcube} suggest that both the square and the cube have
$\sigma_{\mu}$ equal to a single, possibly punctured, interval
$(a,b)\subset[-1,1]$. The square has $a=-0.5$ and $b=0.5$, consistent
with the exact computations of the spectral radius for a square found
in \cite{Kuhn88} and \cite{Wern97}. The cube has $a\approx-0.694526$
and $b=0.5$. For later reference we let $\sigma_{{\mu}{\rm sq}}$
denote $\sigma_{\mu}$ of the square and $\sigma_{{\mu}{\rm cu}}$
denote $\sigma_{\mu}$ of the cube.

For a large class of non-smooth $S$, the potential $U(r)$ exists when
$z$ stays away from a certain compact set $L:\sigma_{\mu}\subset
L\subset[-1,1]$. Furthermore, $\alpha(z)$ has a limit,
\begin{equation}
\alpha^+(x)=\lim_{y\to 0^+}\alpha(x+{\rm i}y)\,, 
\label{eq:alphaplus}
\end{equation}
as $z=x+{\rm i}y$ approaches $x$ from the upper half-plane for almost
all $x\in\mathbb{R}$. See Theorem~\ref{thm:muexist} and
Section~\ref{sec:prop}. It is important in this context and when
$x\in\sigma_{\mu}$ not to interpret $\alpha^+(x)$ as a polarizability
corresponding to a meaningful solution $U(r)$ for a negative
permittivity ratio $\epsilon_2/\epsilon_1=(x-1)/(x+1)$. On the
contrary, Theorem~\ref{thm:Ublowup} states that there is no $U(r)$
with finite energy solving~(\ref{eq:elstat}), (\ref{eq:elstatb1}),
(\ref{eq:elstatb2}) when $z=x\in\sigma_{\mu}$. Therefore, any attempt
to solve the electrostatic problem directly at a point
$z\in\sigma_{\mu}$ is bound to fail.

\subsection{The limit polarizability $\alpha^+(x)$ and its relation to
  $\mu'(x)$}

This paper aims at constructing an efficient scheme for computing
$\alpha(z)$ of a cube at all $z$ for which this quantity exists.
Still, in our numerical experiments we only compute the limit
$\alpha^+(x)$ of~(\ref{eq:alphaplus}) for $x\in[-1,1]$. The reason for
this is that the computation of $\alpha(z)$ is hardest for $z$ close
to $\sigma_{\mu}\subseteq[-1,1]$. Accurate results for $\alpha^+(x)$
therefore indicate a robust scheme. Furthermore, there is a simple
connection between $\alpha^+(x)$ and $\mu'(x)$. Using jump relations
for Cauchy-type integrals one can show from~(\ref{eq:lambda5}) that
\begin{equation}
\mu'(x)=\Im\{\alpha^+(x)\}/\pi\,,\quad x\in\mathbb{R}\,.
\label{eq:herglotz}
\end{equation}

Knowledge of $\mu'(x)$ for non-smooth $S$ is of great interest in
theoretical materials science. Closed form expressions seem to be out
of reach, however, except for a famous example in a periodic
two-dimensional setting~\cite{Crast01,Milt01}. As for numerics, merely
determining $\sigma_{\mu}$ is a challenge~\cite{Perr10,Walle08}. To
the authors' knowledge, $\sigma_{\mu}$ is not known for any $S$
exhibiting corners in three dimensions. The accurate determination of
$\mu'(x)$ is even harder~\cite{Kettu08}. Studying
how~(\ref{eq:lambda4b}) evolves as a smooth $S$ becomes non-smooth is
not an efficient method. Eigenvalue problems are costly to solve. The
discretization of $K$ on surface portions of high curvature is
problematic. Conditioning is also an issue and details of the mapping
$u_iv_i\to \mu'(x)$ need to be worked out. It is desirable to find
$\mu'(x)$ in a more direct way and~(\ref{eq:herglotz}) offers
precisely this. Obtaining $\mu'(x)$ is a subproblem of computing
$\alpha^+(x)$ for $x\in[-1,1]$.

\section{Theory -- preliminaries} 
\label{sec:prelim}

Let $V \subset \R^d$, $d \geq 2$, be an open and bounded set that is
Lipschitz, in the sense that its boundary $S = \partial V$ is
connected and locally the graph of a Lipschitz function in some basis.
For a more precise definition of this concept, see for example
\cite{Verc84}. To avoid a certain technicality we will also assume
that $V$ is star-like, meaning that there exists an $r_0 \in V$ such
that the line segments between $r_0$ and every other point $r \in V$
are contained in $V$. In the applications of this paper, $V$ will take
on the role of the square in $\R^2$, the cube in $\R^3$, or a set with
smooth boundary approximating either of the two. As before we will
denote $E = \overline{V}^c$.

In this section we first record a number of results about the single
and double layer potentials associated with the Laplacian on $V$, to
then introduce the mathematical framework in which we will study the
boundary value problem given by (\ref{eq:elstat}),
(\ref{eq:elstatb1}), (\ref{eq:elstatb2}). Actually, in this section
and the next, we will develop the theory only for $d \geq 3$. The
two-dimensional case contains several anomalies in relation to the
higher-dimensional theory, and it is for the sake of clarity and
brevity that we exclude it. We shall indicate some of the differences
as we progress, but we note here that the main results about the
polarizability $\alpha(z)$ remain true also for $d=2$.

While $L^2(S)$ is a natural domain for the operator $K$, we will
primarily focus on the action of $K$ on certain Sobolev spaces $H^s$.
There is good reason for this. For one, $K$ is not self-adjoint as an
operator on $L^2(S)$, or even normal, so that the spectral theorem can
not be directly applied. We will therefore develop certain
symmetrization techniques, and these demand that we consider $K$ on
fractional Sobolev spaces. A second, related reason, is that the
$L^2$-spectrum of $K$ is no longer contained in the real line when $S$
fails to be smooth, see I. Mitrea \cite{IMitrea02}. We shall see that
considering $K$ on a Sobolev space amends this problem. We also note
here that by the $X$-spectrum of a bounded operator $T$ on a Hilbert
space $X$, $T:X \to X$, we always mean the set
\begin{equation*}
\spec (T, X) = \{z \in \C \, : \, K-z \text{ is not bijective on } X \}.
\end{equation*}
When $z \notin \spec (T,X)$ it is a consequence of the closed graph
theorem that $K-z$ has a bounded inverse $(K-z)^{-1} : X \to X$.

For $s=0$ we have simply that $H^0(V) = L^2(V)$ and $H^0(S) = L^2(S)$.
For $s=1$, $H^1(V)$ is the Hilbert space of distributions $u$ such
that $u$ and $\partial_{x_j} u$, $1 \leq j \leq n$, are members of
$L^2(V)$. The norm is given by
\begin{equation*}
\| u\|_{H^1(V)}^2 =  \|u\|_{L^2(V)}^2 + \|\nabla u\|_{L^2(V)}^2.
\end{equation*}
$H^1(S)$ can be defined similarly using the almost everywhere defined
tangential vectors of $S$, see for example Geymonat \cite{Geym07}. For
$0 < s < 1$, $H^s$ can be defined by real interpolation methods, but
in our situation it can alternatively be characterized by a Besov type
norm. That is, $u \in H^s(V)$ if
\begin{equation*}
\| u \|_{H^s(V)}^2 = \| u \|_{L^2(V)}^2 + \int_{V \times V} \frac{|u(r)-u(r')|^2}{|r-r'|^{d+2s}} \, {\rm d}r \, {\rm d}r' < \infty.
\end{equation*}
$u \in H^s(V)$ is then inductively defined for $s = s_{i} + s_{f}$
with $s_i \geq 1$ an integer and $0 < s_f \leq 1$ by requiring that $u
\in H^{s_i}$ and $\partial^\beta u \in H^{s_f}$ for $\beta = (\beta_1,
\ldots , \beta_d)$ with $\sum \beta_k = s_i$. Returning to the case $0
< s < 1$, we have that $u \in H^s(S)$ if
\begin{equation*}
\| u \|_{H^s(S)}^2 = \| u \|_{L^2(S)}^2 + \int_{S \times S} \frac{|u(r)-u(r')|^2}{|r-r'|^{d-1+2s}} \, {\rm d}\sigma_r \, {\rm d}\sigma_{r'} < \infty,
\end{equation*}
where $\sigma$ denotes Hausdorff measure on $S$. See Adams
\cite{Adams75} and Grisvard \cite{Gris85} for further information and
the equivalence of various definitions in the Lipschitz setting. For
$s > 0$ we define $H^{-s}$ as the dual space of $H^s$ in the
$L^2$-pairing. More precisely, a distribution $u$ lies in $H^{-s}$ if
and only if
\begin{equation*}
\| u\|_{H^{-s}} = \sup_{\|v\|_{H^s}=1} | \langle u, v \rangle_{L^2} | < \infty.
\end{equation*}
We shall also make use of Sobolev traces, which give us a way to
assign boundary values to distributions in $V$. We will only require
the classical Gagliardo result \cite{Gagli57} which says that there
uniquely exists a continuous, surjective linear operator $\Tr : H^1(V)
\to H^{1/2}(S)$ with right continuous inverse such that $\Tr u = u|_S$
for any $u \in C^\infty(\overline{V})$. There is also a corresponding
trace from the exterior domain with the same properties, $\Tr_E :
H^1(E) \to H^{1/2}(S)$.

The $L^2(S)$-adjoint $K^*$ of the operator $K$ is known as the double
layer potential, given by the formula
\begin{equation} \label{eq:dbllayer}
(K^*u)(r) = 2\int_S \frac{\partial}{\partial \nu_{r'}}G(r,r') u(r') \, {\rm d} \sigma_{r'}, \quad u \in L^2(S), \, r \in S.
\end{equation}
For $d=2,3$ the Newtonian kernel $G$ has already been defined in
\eqref{eq:newtker}, and for $d > 3$ it is given by
\begin{equation*}
G(r,r') = \omega_d |r-r'|^{2-d},
\end{equation*}
with a normalization constant $\omega_d$ chosen so that $\Delta_r
G(r,0) = -\delta$ in the sense of distributions. When $S$ is a
$C^2$-surface the kernel of $K^*$ is only weakly singular (for $d=2$
there is no singularity at all present), and it is a standard matter
to see that \eqref{eq:dbllayer} defines $K^*$ as a compact operator on
$H^s(S)$ for $0 \leq s \leq 1$. The compactness of $K^*$ makes its
spectral analysis considerably easier, and in this case it is well
known that
\begin{equation} \label{eq:spec11}
\spec (K^*, L^2(S)) \subset [-1,1), 
\end{equation}
see for example the results of Escauriaza, Fabes and Verchota
\cite{Esca92} together with the fact that the spectrum of $K^*$ is
real in the $C^2$-case. This latter point will be discussed further
later on.

Unfortunately, when $S$ is only a Lipschitz surface, $K^*$ is no
longer compact in general.  In fact, when $S$ is a curvilinear polygon
in two dimensions, I. Mitrea \cite{IMitrea02} has shown that the
$L^2$-spectrum of $K^*$ consists of the union of certain solid
``figure eights'' in the complex plane, one for each (non-smooth)
vertex of $S$, in addition to a finite number of real eigenvalues. In
particular this applies when $S$ is a square in two dimensions, with
only one figure eight present, since all angles are equal. The general
situation is not as well understood, but when $V \subset \R^d$ is
convex, as it is in our situation, it is known that the spectral
radius of $K^*$ on $L^2(S)$ is $1$, see Fabes, Sand and Seo
\cite{Fab92}.

To even define $K^*$ in the general Lipschitz setting, the integral in
\eqref{eq:dbllayer} must be understood in an almost everywhere
principal value sense. We remark, however, that when $S$ is a
curvilinear polyhedron, $K^*u(r)$ can be evaluated in the usual
integral sense, except possibly when $r$ belongs to an edge of $S$,
and so it is not necessary to consider principal values in the main
applications of this paper. Proving the boundedness of $K^*$ on
$L^2(S)$ was an accomplishment of Coifman, McIntosh and Meyer
\cite{Coif82} in their study of singular integrals. The boundedness of
$K^*$ as an operator on $H^s(S)$, $0 < s \leq 1$, also essentially
follows from \cite{Coif82}, see for example Meyer \cite{Meyer90}. By
duality we immediately obtain that $K$ is bounded on $H^{-s}(S)$, $0
\leq s \leq 1$.

For $u \in L^2(S)$ one may of course also evaluate the integral \eqref{eq:dbllayer} in $V \cup E$ to obtain a harmonic function. We denote
\begin{equation*} 
(Du)(r) = 2\int_S \frac{\partial}{\partial \nu_{r'}}G(r,r') u(r') \, {\rm d} \sigma_{r'}, \quad r \in V \cup E.
\end{equation*}
Fabes, Mendez and M. Mitrea \cite{Fab98} prove that for $0 < s < 1$, $D : H^s(S) \to H^{s+1/2}(V)$ is bounded.

The single layer potential of $u$, defined in all of $\R^d$, is given by
\begin{equation*}
(\mathcal{S}u)(r) = 2\int_S G(r,r') u(r') \, {\rm d} \sigma_{r'}, \quad u \in L^2(S), \, r \in \R^d.
\end{equation*}
The kernel $G$ is only weakly singular when $S$ is a Lipschitz surface, so that there is no issue in defining this integral operator. In fact, $\mS$ has smoothening properties. D. Mitrea \cite{DMitrea97} shows that for $0 \leq s \leq 1$, $\mS : H^{-s}(S) \to H^{1-s}(S)$ is a bicontinuous isomorphism and in \cite{Fab98} it is proven that $\mS$ is bounded as a map $\mS : H^{-s}(S) \to H^{\frac{3}{2}-s}(V)$ for $0 < s < 1$. It is clear that $\mS$ is self-adjoint in the $L^2(S)$-pairing and that $\mS u$ is harmonic in $V \cup E$. 

At this stage, a peculiarity of the case $d=2$ appears. In any
dimension, there exists uniquely a function $u_0 \in L^2(S)$ such that
$(I + K)u_0 = 0$ and $\int_S u_0 \, {\rm d} \sigma = 1$, and one can
show that $\mS u_0|_{\overline{V}} \equiv c$ is constant. In higher
dimensions this constant can never be zero, but for $d=2$ there exist
domains such that $c=0$. When this occurs $\mS$ clearly fails to be
injective, and its range is also affected. On the other hand, if $c=0$
for a particular domain $V$, any non-trivial dilation of $V$ will give
a domain with $c \neq 0$ and the properties in the previous paragraph
may be proven to hold for the dilated domain, at least for $s=1/2$,
which will turn out to be the important case for us. See Verchota
\cite{Verc84} for details. Note that since our object of interest, the
polarizability $\alpha(z)$, is scaling invariant, this anomaly of
$\mS$ for $d=2$ presents no real obstacle.

While the kernel of $\mathcal{S}$ is sufficiently nonsingular to
immediately define a continuous function $\mathcal{S}u$ everywhere on
$\R^d$ if $u$ is for example bounded, similar statements are never
true for the kernel of $K^*$. In fact, the following jump formulas
hold for a function $u \in L^2(S)$.
\begin{align}
\mS^{\textrm{int}}u &=\mS^{\textrm{ext}}u=\mS u&
\partial_\nu \mS^{\textrm{int}}u &= u + Ku \nonumber\\
\partial_\nu \mS^{\textrm{ext}}u &= -u +Ku&
D^{\textrm{int}}u &= -u+K^*u \label{eq:jump}\\
D^{\textrm{ext}}u &= u + K^*u, \nonumber
\end{align}
where a superscript {\footnotesize int} or {\footnotesize ext} denotes
taking a limit from the interior or exterior of $S$, respectively. In
general the formulas are true in the sense of non-tangential
convergence almost everywhere on $S$, see \cite{Verc84}. These jump
relations explain why the boundary condition \eqref{eq:elstatb1} leads
to the integral equation \eqref{eq:rep1}. In a moment we shall make a
more precise statement about this. Before doing so, we need to show
that the jump formulas hold in a certain trace sense.

For this purpose, we will also need to consider the Hilbert space $\mathcal{H}(V)/\C$ of harmonic functions $v$ on $V$, modulo constants, with finite energy,
\begin{equation*}
\|v\|^2_{\mathcal{H}(V)/\C} = \int_V |\nabla v|^2 \, {\rm d}r < \infty.
\end{equation*}
Since this semi-norm annihilates constants we consider $v$ and $v + C$, $c \in \C$, to be the same element. Note that $\mathcal{H}(V)/\C$ is continuously contained in $H^1(V)/\C$ by the classical Poincar\'{e} inequality for $V$. Since $V$ is assumed star-like, it is straightforward to use dilations in order to prove that functions which are harmonic and smooth in $\overline{V}$ are dense in $H^1(V)/\C$. In fact, we introduced the hypothesis that $V$ is star-like only to facilitate such density statements.

The Dirichlet problem 
\begin{equation*}
v \in \mathcal{H}(V)/\C \quad \Tr v = u,
\end{equation*}
is well-posed for initial data $u \in H^{1/2}(S)$, see for example \cite{Fab98}. Equivalently, $\Tr : \mathcal{H}(V)/\C \to H^{1/2}(S)/\C$ is a bicontinuous isomorphism. Often we will simply denote $\Tr v = v|_S$ when it is clear what is meant. 

It is established in a paper by Hofmann, Mitrea and Taylor \cite{Hoff10} that Green's formula
\begin{equation} \label{eq:green}
\int_V  \langle \nabla \phi, \nabla \psi \rangle \, {\rm d} r +  \int_V \phi \Delta \psi  \, {\rm d} r = \int_S \phi \partial_\nu \psi \, {\rm d}\sigma
\end{equation}
continues to hold true for $S$ Lipschitz and $\phi, \psi \in C^\infty(\overline{V})$. Since the functions in $\mathcal{H}(V)/\C$ are harmonic, this shows that its scalar product satisfies
\begin{equation*}
\langle v, w \rangle_{\mathcal{H}(V)/\C} = \int_V \langle \nabla v, \nabla \bar{w} \rangle \, {\rm d} r = \int_S v \partial_\nu \bar{w} \, {\rm d}\sigma = \int_S (\partial_\nu v) \bar{w} \, {\rm d}\sigma.
\end{equation*}
Initially these identities are valid only for smooth $v$ and $w$, but as in \cite{Khav91} one can argue by duality and density to interpret the normal derivatives $\partial_\nu v$ and $\partial_\nu w$ as elements of $H^{-1/2}(S)$ so that the equalities remain true. If we denote by $H_0^{-1/2}(S)$ the closed subspace of those $u \in H^{-1/2}(S)$ such that $\int_S u \, {\rm d}\sigma = 0$, the implied duality argument gives rise to a bicontinuous bijective operator $\partial_\nu : H^{1/2}(S)/\C \to H_0^{-1/2}(S)$ which should be understood as the normal derivative of the trace.

It is important for our purposes to now repeat this construction for the exterior space $\mathcal{H}(E)$ of harmonic functions $v$ in $E$ with finite energy norm and $\lim_{r \to \infty} v(r) = 0$. We state this as a proposition. 
\begin{proposition}
The exterior trace is a bicontinuous isomorphism when considered as an operator $\Tr_E : \mathcal{H}(E) \to H^{1/2}(S)$. There is a corresponding bounded bijective operator $\partial_\nu^E : H^{1/2}(S) \to H^{-1/2}(S)$ satisfying
\begin{equation*}
\langle v, w \rangle_{\mathcal{H}(E)} = \int_E \langle \nabla v, \nabla \bar{w} \rangle \, {\rm d} r = -\int_S v \partial_\nu^E \bar{w} \, {\rm d}\sigma = -\int_S (\partial_\nu^E v) \bar{w} \, {\rm d}\sigma.
\end{equation*}
\end{proposition}
\begin{proof}
The statements about $\Tr_E$ follow by the well-posedness of the Dirichlet problem, see \cite{Fab98}. The construction of $\partial_\nu^E$ again follows along the lines of \cite{Khav91}.
\end{proof}
\begin{remark} \label{rmk:2ddirichlet} \rm
When $d=2$ the additional condition $\int_S u \, {\rm d}\sigma = 0$ is required to solve the exterior Dirichlet problem $v \in \mathcal{H}(E), \, \Tr_E v = u$. This is also reflected in the kernel $G(r,r')$ of the single layer potential. Note that $G(r,r') \sim -\frac{1}{2\pi} \log|r|$ as $r \to \infty$ for $d=2$, but $\lim_{r\to\infty} G(r,r') = 0$ for $d > 2$. 
\end{remark}
We end this section with an interpretation of the jump relations \eqref{eq:jump} within the just established framework.
\begin{proposition} \label{prop:tracejump}
Let $u \in H^{1/2}(S)$, then $Du \in \mathcal{H}(V)/\C$, $Du \in \mathcal{H}(E)$ and 
\begin{equation*}
\Tr Du = -u+K^*u \quad \Tr_E Du = u +K^*u.
\end{equation*}
Furthermore, let $v \in H^{-1/2}(S)$. Then $\mS v \in
\mathcal{H}(V)/\C$, $\mS v \in \mathcal{H}(E)$ and
\begin{equation*}
\Tr \mS v = \Tr_E \mS v = \mS v|_S \quad \partial_\nu \mS v = v+Kv \quad \partial_\nu^E \mS v = -v +K v.
\end{equation*}
\end{proposition}
\begin{proof}
Suppose first that $u$ and $v$ are smooth and harmonic on $\overline{V}$. Then by applying Green's formula we find that 
\begin{equation} \label{eq:Dformula}
Du(r) = \mS(\partial_\nu u)(r) - 2u(r), \quad r \in V. 
\end{equation}
In particular, this shows that $Du|_V$ extends continuously to
$\overline{V}$, and hence the jump relation $D^{\textrm{int}}u = -u +
K^* u$ must hold in trace sense. That is, $\Tr Du = -u+K^*u$. Since
both sides of this equation are continuous maps of $H^{1/2}(S)$ by
previously quoted results, it must hold for every $u \in H^{1/2}(S)$.
By the same reasoning we obtain $\Tr \mS v = \mS v|_S$ for every $v
\in H^{-1/2}(S)$.

To show that $\partial_\nu \mS v = v+Kv$ we note that $\mS(\partial_\nu u)  = u+K^*u$ by \eqref{eq:Dformula} and the jump formula. The sought formula is the dual statement of this. More precisely, for any smooth harmonic $\psi$ we have
\begin{align*}
\langle \partial_\nu \mS v, \psi \rangle_{L^2(S)} &= \langle \mS v, \partial_\nu \psi \rangle_{L^2(S)} = \langle v, \mS(\partial_\nu \psi) \rangle_{L^2(S)} \\ &= \langle v, \psi + K^*\psi \rangle_{L^2(S)} = \langle v + Kv, \psi \rangle_{L^2(S)},
\end{align*}
which verifies that $\partial_\nu \mS v = v+Kv$ for all $v \in H^{-1/2}(S)$ by continuity and density. 

The exterior statements are dealt with similarly.
\end{proof}

\section{Theory -- results} \label{sec:results}
\subsection{Existence of the measure $\mu$}
We are now in a position to develop the symmetrization techniques that
have been alluded to previously. Once these are in place, we can use
the spectral theory of self-adjoint operators to prove that the
(scaled) polarizability $\alpha(z) = \frac{\alpha(\epsilon_1,
  \epsilon_2)}{|V|\epsilon_1}$, $z =
\frac{\epsilon_1+\epsilon_2}{\epsilon_1-\epsilon_2} \in \C$, has a
representing measure $\mu$.

We begin, however, by describing the sense in which the potential $U$ will solve the boundary value problem given by \eqref{eq:elstat}, \eqref{eq:elstatb1} and \eqref{eq:elstatb2}. Note that the following proposition furthermore expresses the fact that if looking for a potential such that $U(r) - e \cdot r$ has finite energy, then $H^{-1/2}(S)$ is exactly the right space to find the corresponding density distribution $\rho$.

\begin{proposition} \label{prop:solsense}
Let $\rho \in H^{-1/2}(S)$ be such that \eqref{eq:inteq1} holds, i.e. $(K-z)\rho = g$, where $g(r) = -2(e \cdot \nu_r)$. Let
\begin{equation} \label{eq:Ueq}
U(r) = e \cdot r + \frac{1}{2}\mS \rho(r), \quad r\in \R^d.
\end{equation}
Then $U \in \mathcal{H}(V)/\C$, $U-e \cdot r \in \mathcal{H}(E)$, $\Tr U = \Tr_E U$, $\lim_{r \to \infty} \nabla U = e$ and $U$ satisfies \eqref{eq:elstatb1} in the sense that
\begin{equation} \label{eq:bdrycond}
\epsilon_1 \left ( \partial_\nu^E(U - e\cdot r) + \partial_\nu (e \cdot r) \right) = \epsilon_2  \partial_\nu U.
\end{equation}
The converse is also true. That is, if $U$ satisfies the above properties, then there exists a $\rho \in H^{-1/2}(S)$ such that \eqref{eq:Ueq} and \eqref{eq:inteq1} hold.
\end{proposition}
\begin{proof}
This is a consequence of Proposition \ref{prop:tracejump}, the well-posedness of the interior and exterior Dirichlet problems and the bijectivity of $\mS : H^{-1/2}(S) \to H^{1/2}(S)$. 
\end{proof}
\begin{remark} \rm
When $z \neq -1$, the hypothesis that $(K-z)\rho = g$ implies that $\rho \in H_0^{-1/2}(S)$. This seen by taking into account that $g = -2\partial_\nu (e \cdot r)$ and $\partial_\nu \mS \rho$ both belong to $H^{-1/2}_0(S)$ in the computation
\begin{equation*}
z\int_S \rho \, {\rm d}\sigma = \int_S K\rho - g \, {\rm d}\sigma = \int_S K\rho \, {\rm d}\sigma = \int_S \partial_\nu \mS \rho - \rho \, {\rm d}\sigma = -\int_S \rho \, {\rm d}\sigma.
\end{equation*}
This is of importance for the case $d=2$ (cf. Remark \ref{rmk:2ddirichlet}).
\end{remark}
In the sequel we shall denote $\rho = \rho_z$ and $U = U_z$ to indicate their dependence on $z$.  Under the hypothesis of the preceding proposition, we can, due to the assumption of isotropy, express the scaled polarizability as
\begin{multline*}
\alpha(z) = \frac{\epsilon_2-\epsilon_1}{|V|\epsilon_1} \int_V \nabla U_z(r)  \cdot e\, {\rm d} r =  \frac{\epsilon_2-\epsilon_1}{|V|\epsilon_1}\int_S  (\partial_\nu U_z)(r) (e \cdot r) \, {\rm d}\sigma_r \\ =  \frac{\epsilon_2-\epsilon_1}{|V|\epsilon_1}\int_S  (e \cdot \nu_r + \frac{1}{2}(\rho_z +  K\rho_z)(r))(e \cdot r) \, {\rm d}\sigma_r \\=   \frac{\epsilon_2-\epsilon_1}{|V|\epsilon_1}\frac{z+1}{2} \int_S  \rho_z(r)(e \cdot r) \, {\rm d}\sigma_r = \int_S \rho_z h \, {\rm d}\sigma,
\end{multline*}
where $h(r) = -(e\cdot r)/|V|$. Since $(K-z)\rho_z = g$, this shows that the analysis of $\alpha$ is closely related to the spectral theory of $K$.

The scalar product on $\mathcal{H}(V)/\C$ is one of the keys to understanding the spectral theory of $K$ and $K^*$, since $K^*$ is self-adjoint in the $\mathcal{H}(V)/\C$-pairing. To be precise, the above shows that any element $v \in H^{1/2}(S)/\C$ can also be considered as an element $v \in \mathcal{H}(V)/\C$ in a bicontinuous way. We are therefore justified in letting $H^{1/2}(S)/\C$ inherit its scalar product from $\mathcal{H}(V)/\C$,
\begin{equation*}
\langle v, w \rangle_{H^{1/2}(S)/\C} = \langle v, w \rangle_{\mathcal{H}(V)/\C}.
\end{equation*}
Since $K^*$ maps constants onto constants (cf. \eqref{eq:Dformula}), we may consider $K^*$ as a bounded map on $H^{1/2}(S)/\C$. Let $v,w \in H^{1/2}(S)/\C$. By the fact that $K^* v = \mS(\partial_\nu v) - v$ it then holds that
\begin{align*}
\langle K^*v, w \rangle_{H^{1/2}(S)/\C} &=  \int_S ((\mS(\partial_\nu v) - v)(\partial_\nu \bar{w}) \, {\rm d}\sigma \\ &= \int_S v \partial_\nu(\mS (\partial_\nu \bar{w})-\bar{w}) \, {\rm d}\sigma = \langle v, K^*w \rangle_{H^{1/2}(S)/\C}.
\end{align*}

\begin{theorem} \label{thm:muexist}
There exists a compact set $L \subset \R$ and a positive Borel measure $\mu$ with total mass 2 and compact support contained in $L$, such that for $z \in \C$, $z \notin L$, $(K-z)\rho = g$ has a unique solution $\rho_z \in H_0^{-1/2}(S)$ and
\begin{equation} \label{eq:muexist}
\alpha(z) = \int_\R \frac{{\rm d} \mu(x)}{x-z}.
\end{equation}
$\mu$ is unique in the class of compactly supported finite Borel measures such that \eqref{eq:muexist} holds for all $z \in \C_+ = \{z \, : \, \Im z > 0 \}$.

\end{theorem}
\begin{proof}
Since $K^*$ is bounded and self-adjoint on $H^{1/2}(S)/\C$ it has by the spectral theorem a corresponding projection-valued spectral measure $\mathcal{E}$ with support in the spectrum of $K^*$. Let $L = \spec K^*$. $\mathcal{E}$ is characterized by the fact that for every bounded Borel-measurable function $f$ on $L$, it holds that
\begin{equation} \label{eq:specteq}
f(K^*) = \int_L f(x) \, {\rm d}\mathcal{E}(x).
\end{equation}
For any $z \notin L$ note that $K^*-\bar{z}$ is invertible on $H^{1/2}(S)/\C$, and hence $K-z$ is invertible on the dual space $H_0^{-1/2}(S)$. Since $h \in H^{1/2}(S)$ and $g =  2|V|\partial_\nu h \in H_0^{-1/2}(S)$ we have 
\begin{align*}
\alpha(z) &= \langle (K-z)^{-1} g, h \rangle_{L^2(S)} =  \langle g, (K^* - \bar{z})^{-1} h \rangle_{L^2(S)} \\
&= 2|V| \langle \partial_\nu h, (K^* - \bar{z})^{-1} h \rangle_{L^2(S)} = 2|V| \langle h, (K^* - \bar{z})^{-1} h \rangle_{H^{1/2}(S)/\C} \\ &= 2|V| \langle (K^* - z)^{-1} h,  h \rangle_{H^{1/2}(S)/\C}.
\end{align*}
With $\mu(R) = 2|V| \langle \mathcal{E}(R) h , h \rangle_{H^{1/2}(S)/\C}$ for any Borel set $R \subset \R$ we obtain by applying \eqref{eq:specteq} with $f(x) = 1/(x-z)$ the desired formula
\begin{equation*}
\alpha(z) = \int_\R \frac{{\rm d} \mu(x)}{x-z}.
\end{equation*}
The constructed measure $\mu$ is positive since $\mu(R) = 2|V| \langle \mathcal{E}(R) h , h \rangle = 2|V| \langle \mathcal{E}(R) h , \mathcal{E}(R)h \rangle = 2|V| \|\mathcal{E}(R) h \|^2 \geq 0$, and since $\mathcal{E}(\R)$ is the identity operator, its total mass is
\begin{equation*}
\| \mu \| = 2|V|\| h\|^2 = 2\frac{1}{|V|}\int_V |e|^2 \, {\rm d} r = 2.
\end{equation*}
Equation \eqref{eq:muexist} expresses that $\alpha$ is the Cauchy transform of $\mu$, $\alpha(z) = \mathcal{K}_\mu(z)$. See Cima, Matheson and Ross \cite{Cima06} for an excellent survey of the Cauchy transform, written for measures with support on the unit circle, but all results can be transformed to results about measures on $\R$ through standard conformal mapping techniques. See Koosis \cite{Koos98} for an explanation of this latter point, as well as results stated directly for the real line. By the classical F. and M. Riesz theorem, any other measure $\mu^*$ supported on $\R$ and such that $\mathcal{K}_{\mu^*}(z) = \alpha(z) = \mathcal{K}_\mu(z)$ for $z \in \C_+$ must be of the form ${\rm d}\mu^* = {\rm d}\mu + \bar{u} \, {\rm d }x$, where $u \in \mathcal{H}^1(\C_+)$ is given by the boundary values of an analytic Hardy space-function in the upper half-plane. Since such a function $u$ can never be compactly supported unless it is identically zero, we obtain the uniqueness part of the theorem.
\end{proof}
\begin{remark}\rm
  The fact that $\| \mu \| = 2$ is sum rule \eqref{eq:sum1}. Note also
  that $L=\spec(K^*,H^{1/2}/\C)=\spec(K,H_0^{-1/2})$ only differs from
  $\spec(K^*,H^{1/2})=\spec(K,H^{-1/2})$ by the point $z=-1$, which is
  in the latter spectrum but not in the former, see~\cite{DMitrea97}.
\end{remark}
The considerations leading up to the previous theorem are quite
similar in spirit to those of Bergman \cite{Berg79}. Bergman considers
a different potential operator which is symmetric under the inner
product $\int_V \langle \nabla v, \nabla w \rangle \, {\rm d}r$,
although the arguments presented are somewhat incomplete since a space
of functions belonging to this inner product is not identified.

Before discussing the specific features of $\alpha$ and $\mu$ we shall
gain some further insight into the structure of $K$ and $K^*$ by
investigating a different symmetrization approach, expounded upon in
the case when $S$ is a $C^2$-surface by Khavinson, Putinar and Shapiro
\cite{Khav91}. The starting point is Plemelj's symmetrization
principle, which says that $\mS K = K^* \mS$ on $L^2(S)$ and continues
to hold true in our situation with the same proof as in \cite{Khav91}.
This operator equality amounts to the statement that $K$ is
self-adjoint under the inner product $\langle \mS u, v
\rangle_{L^2(S)}$. Note that $\mS$ is a strictly positive operator on
$L^2(S)$, so that the form $(u,v) \to \langle \mS u, v
\rangle_{L^2(S)}$ is strictly positive definite.

Instead of working with the completion of $L^2(S)$ under $\langle \mS u, u \rangle_{L^2(S)}$, we shall follow the approach of \cite{Khav91} and introduce an operator-theoretic formalism to express the symmetrization of $K$. Recall that $K: L^2(S) \to L^2(S)$ is compact when $S$ is $C^2$, and so its spectrum consists of the point $0$ and a sequence $(z_i)$ of non-zero eigenvalues tending to zero, every eigenspace $H_{z}(K) = \ker (K-z)$ having finite dimension when $z \neq 0$. 
\begin{theorem}[\cite{Khav91}] \label{thm:Khav}
Suppose that $S$ is a $C^2$-surface. Then there exists a self-adjoint compact operator $A : L^2(S) \to L^2(S)$ such that $A\sqrt{\mS} = \sqrt{\mS} K$ on $L^2(S)$.  When $z \neq 0$, $\sqrt{\mS} : H_{z}(K) \to H_{z}(A)$ and $\sqrt{\mS} : H_{z}(A) \to H_{z}(K^*)$ are isomorphisms of the indicated eigenspaces, and the eigenvectors of $K^*$ (including those for $z = 0$) span $L^2(S)$. In particular, the $L^2(S)$-spectrum of $K^*$ is real.  
\end{theorem}
From the point of view of Khavinson et al., the proof of this theorem
parallels the symmetrization theory of Krein \cite{Krein}, which mainly concerns compact operators. However, the theory of Hassi, Sebesty\'{e}n
and de Snoo \cite{Hass05} assures us of the existence of $A$ even when $K$ is only bounded, that is, when $S$ is only Lipschitz. The key fact is still
that $\mS K = K^* \mS$.
\begin{proposition}
There exists a bounded self-adjoint operator $A: L^2(S) \to L^2(S)$ such that $A \sqrt{\mS} = \sqrt{\mS} K$ on $L^2(S)$.
\end{proposition}
\begin{proof}
Given the existence of a bounded operator $A$ on $L^2(S)$ satisfying $A \sqrt{\mS} = \sqrt{\mS} K$, we deduce from the equation $\sqrt{\mS}A\sqrt{\mS} = \mS K = K^* \mS = \sqrt{\mS}A^*\sqrt{\mS}$, and the injectivity and dense range of $\mS$, that $A$ must in fact be self-adjoint.
\end{proof}
In fact, we have the following result, contained in \cite[Proposition
1]{Khav91} with a proof that carries over to our setting.
\begin{proposition}
$\sqrt{\mS} : L^2(S) \to H^{1/2}(S)$ is a bicontinuous isomorphism. Dually, $\sqrt{\mS}$ also extends to a bicontinuous isomorphism $\sqrt{\mS} : H^{-1/2}(S) \to L^2(S)$.
\end{proposition}
The existence of $A$ gives us another way to derive the existence of a measure $\mu$ such that $\alpha(z) = \int_\R \frac{{\rm d} \mu(x)}{x-z}$. Since 
\begin{equation} \label{eq:Kresolv}
(K-z)^{-1} = \sqrt{\mS}^{-1}(A-z)^{-1} \sqrt{\mS}
\end{equation}
 as an inverse on $H^{-1/2}(S)$, we can write
\begin{equation} \label{eq:Aspecrep}
\alpha(z) = \langle (K-z)^{-1}g, h \rangle_{L^2(S)} = \langle (A-z)^{-1} \sqrt{\mS} g, \sqrt{\mS}^{-1} h \rangle_{L^2(S)},
\end{equation}
where, clearly, $\sqrt{\mS} g, \sqrt{\mS}^{-1} h \in L^2(S)$. If $\mathcal{E}_A$ is the spectral measure of $A$, we therefore obtain a measure with the desired property by letting $\mu(R) = \langle \mathcal{E}_A(R) \sqrt{\mS} g, \sqrt{\mS}^{-1} h \rangle_{L^2(S)}$ for $R \subset \R$. In this way we immediately recover all the conclusions of Theorem \ref{thm:muexist}, except for the positivity of $\mu$. 

To see the positivity, recall for $u \in H^{1/2}(S)$ that $\mS(\partial_\nu u) = u + K^*u$ and consider the computation
\begin{equation*}
\sqrt{\mS}(\partial_\nu u) = \sqrt{\mS}^{-1}(u+K^*u) = (I+A)\sqrt{\mS}^{-1}u.
\end{equation*}
In view of the identity $g = 2|V|\partial_\nu h$ we find that
\begin{equation*}
\alpha(z) = 2|V| \langle (A+I)(A-z)^{-1} \sqrt{\mS}^{-1} h, \sqrt{\mS}^{-1} h \rangle_{L^2(S)}.
\end{equation*}
Clearly, it follows that $\mu \geq 0$ if $A+I \geq 0$, that is, if 
\begin{equation} \label{eq:poscond}
\spec(K, H^{-1/2}(S)) = \spec(A,L^2(S)) \subset [-1,\infty).
\end{equation}
In the $C^2$-case we know from \eqref{eq:spec11} and the compactness of $K$ that $\spec(K, H^{-1/2}) \subset [-1,1]$. We shall see in Theorem \ref{thm:wstarconv} that this continues to be true when $S$ is merely a Lipschitz surface. 
\subsection{Properties of the polarizability $\alpha(z)$ and the measure $\mu$}
\label{sec:prop}

In the final part of this section we will employ all the tools
introduced thus far, in order to understand the behavior and specific
features of the polarizability $\alpha(z)$ and its related measure
$\mu$. We begin by discussing the case of a smooth surface $S$, giving
a rigorous treatment of many of the formulas and ideas that appear in
the paper \cite{Mayer05} by Mayergoyz, Fredkin and Zhang. From there
we discuss the idea of approximating a non-smooth surface by a
sequence of smooth surfaces, proving that the corresponding
representing measures converge in a weak-star sense. We conclude with
an analysis of the measure $\mu$ in the general non-smooth case, in
particular giving a condition which guarantees that it does not have a
singular part. Furthermore, we show that even though the
polarizability $\alpha^+(x)$ exists in a limit sense almost everywhere
$x \in \R$, the potential $U_x$ does not exist at points where
$\mu'(x) > 0$, neither as a solution of the boundary value problem
given by \eqref{eq:elstat}, \eqref{eq:elstatb1} and
\eqref{eq:elstatb2}, nor in a limit sense.

Suppose now that $S$ is a $C^2$-surface, so that the operator $A$ is compact. Let $(f^0_i)$ be an orthonormal basis for $\ker A$, and let $(f^1_i)$ be eigenvectors of $A$ corresponding to the non-zero eigenvalues $(z_i^1)$, repeated according to multiplicity, so that $(f_i) = (f^0_i) \cup (f^1_i)$ is an orthonormal basis for $L^2(S)$. By Theorem \ref{thm:Khav}, $K,K^* : L^2(S) \to L^2(S)$ have the same non-zero eigenvalues $(z_i^1)$, and corresponding (non-orthonormal) eigenvectors are obtained as $\phi^1_i = \sqrt{\mS}^{-1} f^1_i$ and $\psi^1_i = \sqrt{\mS} f^1_i$, respectively. Concerning zero eigenvectors, $\psi_i^0 = \sqrt{\mS}f_i^0 \in H^{1/2}(S)$ are clearly in the kernel of $K^*$, but $\phi_i^0 = \sqrt{\mS}^{-1}f_i^0$ are in general elements of $H^{-1/2}(S)$ and only zero eigenvectors of $K$ when considered as an operator on said space $H^{-1/2}(S)$. In particular, we note that the $H^{-1/2}(S)$-eigenvectors of $K$ span the whole space $H^{-1/2}(S)$.  With this in mind, we shall denote $(\phi_i) = (\phi_i^0) \cup (\phi_i^1)$ and $(\psi_i) = (\psi_i^0) \cup (\psi_i^1)$, the indexing arranged appropriately so that $$\langle \phi_i, \psi_j \rangle_{L^2(S)} = \delta_{ij}.$$ As a final notational detail, we shall let $(z_i)$ denote the full sequence of eigenvalues including zeros, so that $K \phi_i = z_i \phi_i$ and $K^* \psi_i = z_i \psi_i$ for every $i$.

Based on the spectral decomposition of $A$ we obtain for $u \in L^2(S)$ that
\begin{equation*}
\sqrt{\mS}K u = A \sqrt{\mS} u = \sum_{i} z_i \langle \sqrt{\mS} u, f_i \rangle_{L^2(S)} f_i,
\end{equation*}
with convergence in $L^2(S)$. Equivalently, we have for $u \in H^{-1/2}(S)$ the expansion
\begin{equation} \label{eq:kernelformal}
K u = \sum_{i} z_i \langle u, \psi_i \rangle_{L^2(S)} \phi_i,
\end{equation}
with convergence in $H^{-1/2}(S)$. This is the formal interpretation
of \eqref{eq:kernel}. In this framework it is now easy to furthermore
justify \eqref{eq:lambda4a}, \eqref{eq:lambda4b} and
\eqref{eq:resid1}. We state this as a theorem.
\begin{theorem} \label{thm:smooth}
Let $S$ be a $C^2$-surface. Then \eqref{eq:kernelformal} holds with $K$ considered as an operator on $H^{-1/2}(S)$. Furthermore, let $u_i = \langle \phi_i, h\rangle_{L^2(S)}$ and $v_i = \langle g, \psi_i \rangle_{L^2(S)}$. Then
\begin{equation} \label{eq:specequal}
 \spec (K, L^2(S)) = \spec (K,H^{-1/2}(S)) = \{(z_i)\},
\end{equation}
 and for $z \notin (z_i)$ the unique solution $\rho_z \in L^2(S)$ of $(K-z)\rho = g$ is given by
\begin{equation}\label{eq:rhodisc}
\rho_z = \sum_i \frac{v_i \phi_i}{z_i-z},
\end{equation}
with convergence in $H^{-1/2}(S)$. The corresponding formula for the
scaled polarizability is
\begin{equation} \label{eq:alphadisc}
\alpha(z) = \sum_i \frac{u_i v_i}{z_i-z},
\end{equation}
where the sum is absolutely convergent. Finally, it holds that $\sum_i u_i v_i = 2$. 
\end{theorem}
\begin{proof}
Equation \eqref{eq:Kresolv} and Theorem \ref{thm:Khav} show that
\begin{equation*}
\spec (K,H^{-1/2}(S)) = \spec (A, L^2(S)) = \spec (K, L^2(S)).
\end{equation*}
Formula \eqref{eq:rhodisc} also follows from \eqref{eq:Kresolv} and spectral decomposition of $A$, since
\begin{equation*}
\rho_z = \sqrt{\mS}^{-1}(A-z)^{-1} \sqrt{\mS} g = \sqrt{\mS}^{-1} \sum_i \frac{\langle \sqrt{\mS} g, f_i \rangle_{L^2(S)}f_i}{z_i-z} = \sum_i \frac{v_i \phi_i}{z_i-z}.
\end{equation*}
From $\alpha(z) = \langle \rho_z, h \rangle_{L^2(S)}$ we now deduce \eqref{eq:alphadisc}. In terms of the measure $\mu$, this expresses the fact that $\mu = \sum_i u_i v_i \delta_{z_i}$ where $\delta_{z_i}$ is the Dirac delta at $z_i$. This makes it clear that we have already proven that $\sum_i u_i v_i = 2$ as a part of Theorem \ref{thm:muexist}.
\end{proof}
\begin{remark} \rm
The statement that the $H^{-1/2}(S)$-spectrum of $K$ is equal to its $L^2(S)$-spectrum is markedly untrue when $S$ is not $C^2$. In fact, when $S$ is a square in two dimensions the $H^{-1/2}(S)$-spectrum of $K$ is real, while the $L^2(S)$-spectrum extends into the complex plane.
\end{remark}

We will now briefly discuss the limiting process which occurs when
$(S_k)$ is a sequence of smooth surfaces approximating a non-smooth
surface $S$. For brevity we assume that $S$ is the cube defined by
$\max_{1 \leq i \leq d} |r_i| = 1$ in $d$ dimensions, $r = (r_1,
\ldots, r_d)$, and that $S_k$ is the superellipsoid defined by
$\sum_{i = 1}^d |r_i|^k = 1$. We shall also skip some rather laborious
details. See \cite{Verc84} for detailed approximation arguments when
$S$ is a general Lipschitz surface.

For any notation introduced so far, we denote by a subscript $k$ that
it corresponds to the surface $S_k$ rather than $S$. Via appropriately
defined homeomorphisms of $S_k$ onto $S$ we may, in a bicontinuous
way, consider $K_k^*$ and $K^*$ to be operators on the same space
$H^{1/2}(S)$. We choose the implied isomorphism between $H^{1/2}(S)$
and $H^{1/2}(S_k)$ so that it extends to a unitary map of $L^2(S)$
onto $L^2(S_k)$. Similar conventions will apply throughout what
follows. Using the boundedness results collected in Section
\ref{sec:prelim} and adapting the arguments in the proof of
\cite[Theorem 3.1]{Verc84}, one can verify that $K_k^*$ converges
strongly to $K^*$ on $H^{1/2}(S)$, meaning that $\lim_k K_k^* u =
K^*u$ in norm for any $u \in H^{1/2}(S)$. Similarly, $K_k$ converges
strongly to $K$ on $H^{-1/2}(S)$.

Let $z \in \C$ be a point at a distance at least $\varepsilon > 0$
away from the $H^{-1/2}(S)$-spectra of $K$ and $K_k$, for all $k$. It
is evident from \eqref{eq:rhodisc} that $ \sup_{k} \| \rho_{z,k}
\|_{H^{-1/2}(S)} < \infty$. We can hence extract a weakly convergent
subsequence $\rho_{z,k'}$ with limit $b$,
\begin{equation*}
\lim_{k' \to \infty} \langle \rho_{z,k'}, w \rangle_{L^2(S)} = \langle b, w \rangle_{L^2(S)}, \quad \forall w \in H^{1/2}(S).
\end{equation*}
Since $s-\lim K_k = K$ we find that $(K_{k'}-z)\rho_{z,k'}$ converges
weakly to $(K-z)b$. On the other hand, $(K_{k'}-z)\rho_{z,k'} =
g_{k'}$ converges to $g$. Hence $b = \rho_z$. Since every weakly
convergent subsequence of $\rho_{z,k}$ has the same limit $ \rho_z$,
we conclude that $\rho_{z,k}$ converges weakly to $\rho_z$. In
particular, noting that $h_k \to h$ in $H^{1/2}(S)$, $$\lim_{k \to
  \infty} \alpha_k(z) = \lim_{k \to \infty} \langle \rho_{z,k}, h_k
\rangle_{L^2(S)} = \alpha(z).$$

Since this argument remains correct no matter the choice of $g \in
H^{-1/2}(S)$ and $h \in H^{1/2}(S)$, it is not difficult to see, based
on the spectral representation \eqref{eq:Aspecrep}, that any $z \in
\spec (K, H^{-1/2})$ is obtained as the limit of a sequence of
eigenvalues $z_k \in \spec (K_k, H^{-1/2})$. See Weidmann
\cite{Weid97}. From \eqref{eq:spec11} and \eqref{eq:specequal} we
deduce that
\begin{equation*}
\spec (K, H^{-1/2}(S)) = \spec (K^*, H^{1/2}(S)) \subset [-1,1].
\end{equation*}
However, unlike the case of convergence in operator norm, we point out
that strong convergence allows us to say very little about the
character of the points in $\spec (K,H^{-1/2})$. For example, they do
not need to be eigenvalues.

We summarize.

\begin{theorem} \label{thm:wstarconv}
  Let $S$ be the unit cube in $\R^d$, and let $S_k$ be the
  approximating superellipsoid given by $\sum_{i = 1}^d |r_i|^k = 1$,
  $r = (r_1, \ldots, r_d)$. Denoting by ${\rm cl } \, R$ the closure
  of $R \subset \R$, let
\begin{equation*}
B = {\rm cl} \left[\bigcup_k \spec (K_k ,H^{-1/2}(S_k))\right].
\end{equation*}
Then $\spec (K, H^{-1/2}(S)) \subset B \subset [-1,1]$ and $\lim_k
\alpha_k(z) = \alpha(z)$ for all $z \notin B$. Furthermore, the
supports of $\mu_k$ and $\mu$ are contained in $B$, and $\mu_k$
converges weak-star to $\mu$ as measures on $B$. That is,
\begin{equation} \label{eq:weakstar}
\lim_{k\to\infty} \int_\R f(x) \, {\rm d} \mu_k(x) = \int_\R f(x) \, {\rm d} \mu(x), \quad \forall f \in C(B),
\end{equation}
where $C(B)$ denotes the continuous functions on $B$.
\end{theorem}
\begin{proof}
  Only the final point remains to be proven. However, the argument
  preceding the theorem can be repeated to show that
  \eqref{eq:weakstar} holds for $f(x) = 1/(x-z)^t$, where $z \notin B$
  and $t \geq 0$ is an integer. Now \eqref{eq:weakstar} immediately
  follows in general from the fact that the linear span of such
  functions $1/(x-z)^t$ is dense in $C(B)$, which in turn follows for
  example from the Stone-Weierstrass theorem.
\end{proof}

We turn to the discussion of $\mu$ and $\alpha$ for a general
Lipschitz surface $S$. First note that the support of $\mu$ is
contained in the $H^{-1/2}(S)$-spectrum of $K$, and that $\mu$ has a
unique decomposition
\begin{equation*}
\mu = \mu_a + \mu_p + \mu_s.
\end{equation*}
Here $\mu_a$ is the absolutely continuous part of the measure, so that $\mu_a = \mu'(x) \, {\rm d} x$, where $\mu' \in L^1(\R)$. $\mu$ has at most a countable number of atoms $z_i$, corresponding to bright plasmons at $z_i$, and $\mu_p$ is the atomic part of $\mu$, $$\mu_p = \sum_i \mu(\{z_i\}) \delta_{z_i}.$$ Note that each atom arises from an eigenvalue $z$ of $K$, but not necessarily every eigenvalue is given positive measure by $\mu$, reflecting the distinction between bright and dark plasmons. Finally, $\mu_s$ denotes the singular part of $\mu$, excluding atoms. This means that $\mu_s$ has no atoms, yet lives solely on a set with Lebesgue measure zero. That is, there exists a measure zero set $R_0$ such that $\mu_s(R) = 0$ for any Borel set $R \subset \R$ with $R \cap R_0 = \emptyset$.

For $0 < p < \infty$, let $\mathcal{H}_{{\rm conf}}^p(\C_+)$ denote
the conformally invariant Hardy space on the upper half-plane $\C_+$,
consisting of analytic functions $f$ in $\C_+$ such that
\begin{equation*}
\|f\|_p^p = \sup_{r > 0} \int_{y = r} \frac{|f(z)|^p}{|z+i|^2} \, {\rm d} x < \infty, \quad z = x+iy.  
\end{equation*}
This is known as the conformally invariant Hardy space as
$$\mathcal{H}_{{\rm conf}}^p(\C_+) = \{ f \circ \omega \, : \, f \in
H^p(\mathbb{D}) \},$$
where $H^p(\mathbb{D})$ is the usual Hardy space
of the unit disk and $\omega$ is a conformal map of $\C_+$ onto the
disk. $\mathcal{H}_{{\rm conf}}^p(\C_+)$ does not coincide with the
usual Hardy space of the upper half-plane. Of importance to us is the
fact that every $f \in \mathcal{H}_{{\rm conf}}^p(\C_+)$ has boundary
values $f(x) = \lim_{y \to 0^+} f(x+iy)$ almost everywhere, and
\begin{equation*}
\|f\|_p^p = \int_\R \frac{|f(x)|^p}{1+x^2} \, {\rm d} x.
\end{equation*}
See Koosis \cite{Koos98} for further information.

As previously mentioned in the proof of Theorem \ref{thm:muexist},
$\alpha$ is the Cauchy transform of $\mu$, $\alpha(z) =
\mathcal{K}_\mu(z)$. Changing variables in the Cauchy transform
through $\omega$, exploiting that $\mu$ is compactly supported and
applying Smirnov's theorem, see \cite{Cima06}, leads to the fact that
$\alpha \in \mathcal{H}_{{\rm conf}}^p(\C_+)$ for $0 < p < 1$. Hence,
$\alpha$ has boundary values $\alpha^+(x)$ for almost all $x \in \R$,
taken as limits from the upper half-plane. A similar discussion could
be carried out for the lower half-plane, but since $\mu$ is real, the
boundary values obtained from below are related to those obtained from
above simply by conjugation.

The absolutely continuous part of $\mu$ is related to $\alpha^+$ through equation \eqref{eq:herglotz} almost everywhere, $\pi \mu'(x) = \Im (\alpha^+(x))$. Atoms $z_i$ can be recognized as those points where $|\alpha(z_i + i\varepsilon)| \sim \varepsilon^{ -1}$ as $\varepsilon \to 0^+$. While some results are available, see \cite{Cima06}, recovering $\mu_s$ from $\alpha$ is much more subtle. However, if it turns out to be the case that 
\begin{equation} \label{eq:alphahpcond}
\int_\R \frac{|\alpha^+(x)|}{1+x^2} \, {\rm d} x < \infty,
\end{equation}
it follows that $\alpha \in \mathcal{H}^1_{{\rm conf}}(\C_+)$. But then $f \circ \omega^{-1}$ is the Cauchy integral of its boundary values and we infer that $\mu$ must be absolutely continuous. That is, in this case $\mu_p = \mu_s = 0$ and 
\begin{equation} \label{eq:alpharep}
\alpha(z) = \int_\R \frac{\mu'(x) \, {\rm d}x}{x-z}, 
\quad z \notin \spec (K, H_0^{-1/2}(S)).
\end{equation}
While we offer no strict proof, the numerical evidence in Sections
\ref{sec:numsquare} and \ref{sec:numcube} suggests that
\eqref{eq:alphahpcond} holds when $S$ is a square in two dimensions or
a cube in three, and hence that \eqref{eq:alpharep} holds.

To end this section we shall show that while $\alpha(z)$ has boundary values $\alpha^+(x)$ almost everywhere, the same cannot be said of the corresponding solutions $\rho_z$. Note that the limit $\lim_{h\to0} \int_x^{x+h} \, {\rm d}\mu $ exists finitely almost everywhere for $x \in \R$. We say that $\mu'(x)$ exists whenever this is so, $$\mu'(x) = \lim_{h\to0} \int_x^{x+h} \, {\rm d}\mu,$$ and as before we denote by $\sigma_\mu$ the set
\begin{equation*}
\sigma_\mu = \{ x \in \R \, : \, \mu'(x) > 0 \textrm{ exists} \}.
\end{equation*}
\begin{theorem} \label{thm:Ublowup}
  For $x \in \R$ and $\varepsilon > 0$, let $\rho_{x+i\varepsilon} \in
  H^{-1/2}(S)$ denote the unique solution of
  $$(K-x-i\varepsilon)\rho_{x+i\varepsilon} = g.$$ For any $x \in
  \sigma_\mu$, we have
\begin{equation} \label{eq:rhoblowup}
\lim_{\varepsilon \to 0^+} \| \rho_{x+i\varepsilon} \|_{H^{-1/2}(S)} = \infty,
\end{equation}
and moreover, the equation $(K-x) \rho = g$ has no solution $\rho \in H^{-1/2}(S)$. Hence, for such $x$, there exists no potential $U_x$ solving the boundary value problem \eqref{eq:elstat}, \eqref{eq:elstatb1}, \eqref{eq:elstatb2} in the sense given by Proposition \ref{prop:solsense}.
\end{theorem}
\begin{proof}
If \eqref{eq:rhoblowup} does not hold, there is a sequence $(\varepsilon_n)$ with $\varepsilon_n \to 0$ as $n \to \infty$ and such that $\lim_{n\to\infty}\rho_{x+i\varepsilon_n}$ exists weakly. From the proof of Theorem \ref{thm:muexist} we see that $\tau_{x+i\varepsilon_n} = (K^*-x-i\varepsilon_n)^{-1}h$ converges weakly in $H^{1/2}(S)/\C$ to some element $\tau_x$, meaning that
\begin{equation*}
\lim_{n \to \infty} \langle \tau_{x+i\varepsilon_n}, w \rangle_{H^{1/2}(S)/\C} = \langle \tau_{x}, w \rangle_{H^{1/2}(S)/\C}, \quad \forall w \in H^{1/2}(S)/\C.
\end{equation*}
Then, clearly, it must be that 
\begin{equation} \label{eq:Kstarsolve}
(K^*-x) \tau_x = h.
\end{equation}
 Let $\mu_{\tau_x}$ be the positive measure defined by $\mu_{\tau_x}(t) = \langle \mathcal{E}(t) \tau_x, \tau_x \rangle_{H^{1/2}(S)/\C}$, where $\mathcal{E}$ is the spectral measure of $K^*$ on $H^{1/2}(S)/\C$.  Recalling that $\mu(t) = 2|V|\langle \mathcal{E}(t) h, h \rangle_{H^{1/2}(S)/\C}$, we can reformulate \eqref{eq:Kstarsolve} as $$\mu(t) = 2|V|(t-x)^2\mu_{\tau_x}(t).$$ We infer
\begin{equation*}
\int_\R \frac{{\rm d}\mu(t)}{(t-x)^2} < \infty. 
\end{equation*}
This implies that $\lim_{y \to 0^+} P_\mu(x+iy) = 0$ for the Poisson integral
\begin{equation*}
P_\mu(z) = \frac{1}{\pi} \int_\R \frac{y\,{\rm d}\mu(t)}{(t-x)^2+y^2}, 
\quad z = x+iy \in \C_+ .
\end{equation*}
On the other hand, it is well known that $\lim_{y \to 0^+} P_\mu(x+iy) = \mu'(x)$ whenever the right hand side exists finitely. This contradicts the hypothesis that $\mu'(x) \neq 0$.

Suppose that there exists a $\rho \in H^{-1/2}(S)$ such that $(K-x)\rho = g$. Note that
\begin{equation*}
(A-x-i\varepsilon)^{-1}(A-x) = \int_\R \frac{t-x}{t-x-i\varepsilon} \, {\rm d} \mathcal{E}_A(t),
\end{equation*}
where $\mathcal{E}_A$ is the spectral measure of $A$ on $L^2(S)$. Hence the operator norms $\|(A-x-i\varepsilon)^{-1}(A-x)\|$ are uniformly bounded for $\varepsilon > 0$. Since
\begin{equation*}
\rho_{x+i\varepsilon} = (K-x-i\varepsilon)^{-1}(K-x)\rho = \sqrt{\mS}^{-1}(A-x-i\varepsilon)^{-1}(A-x)\sqrt{\mS} \rho
\end{equation*}
we have obtained a contradiction to \eqref{eq:rhoblowup}.
\end{proof}
\begin{remark} \rm
As to be expected, a similar argument shows that the conclusions of Theorem \ref{thm:Ublowup} are valid also when $x \in \R$ is such that $\mu(\{x\}) > 0$, that is, when there is a bright plasmon at $x$.
\end{remark}
Nevertheless, $x \mapsto \rho_x$ for $x \in \R$ can always be given an
interpretation as a $H^{-1/2}(S)$-valued distribution, for example in
the following weak sense. Let $u: \R \to \C$ be a
$C_0^\infty$-function, let $w \in H^{1/2}(S)$, and consider the
functional
\begin{equation*}
\rho^+(u,w) = \lim_{\varepsilon \to 0^+} \int_\R u(x) \langle \rho_{x+i\varepsilon}, w \rangle_{L^2(S)} \, {\rm d} x.
\end{equation*}
Observe that
\begin{equation*}
\int_\R u(x) \langle \rho_{x+i\varepsilon}, w \rangle_{L^2(S)} \, {\rm d} x = \int_\R u(x) \int_\R \frac{{\rm d} \mu_{g,w}(t)}{t-x-i\varepsilon} \, {\rm d} x = -\int_\R \mathcal{K}_u(t-i\varepsilon) \, {\rm d} \mu_{g,w}(t),
\end{equation*}
where $\mu_{g,w}(t) = \langle \mathcal{E}_A(t) \sqrt{\mS}g, \sqrt{\mS}^{-1}w \rangle_{L^2(S)}$. Due to the high regularity of $u$, passing to the limit in the Cauchy transform $\mathcal{K}_u(t-i\varepsilon)$ poses no problems, and we obtain
\begin{equation*}
\rho^+(u,w) = -\int_\R \mathcal{K}_u(t) \, {\rm d} \mu_{g,w}(t).
\end{equation*}
Furthermore, $\rho^+$ solves the equation $(K-x)\rho^+ = g$ in the sense that
\begin{equation*}
\rho^+(u(x),K^*w) - \rho^+(xu(x),w) = \langle g , w \rangle_{L^2(S)} \int_\R u(x) \, {\rm d} x.
\end{equation*}
Of course, we could equally well have introduced a solution $\rho^-$,
\begin{equation*}
\rho^-(u,w) = \lim_{\varepsilon \to 0^+} \int_\R u(x) \langle \rho_{x-i\varepsilon}, w \rangle_{L^2(S)} \, {\rm d} x.
\end{equation*}
We point out also that $\Re \rho = (\rho^+ + \rho^-)/2$ and $\Im \rho = -i(\rho^+ - \rho^-)/2$ give us two real-valued distributions satisfying $(K-x) \Re \rho = g$ and $(K-x) \Im \rho = 0$.

\section{Capacitance}
\label{sec:capa}

The electrical capacitance $C$ of an isolated charged conducting body
$V$ can be defined as the ratio of its total charge to its constant
potential $U(r)$. The problem of calculating $C$ for $V$ being a
(unit) cube has ``long been considered one of the major unsolved
problems of electrostatic theory''~\cite{Read04} and attracted
interest by researchers in computational electromagnetics for half a
century. See~\cite{Bontz11,Hwang10,Mukh09} for recent contributions
along with reviews of previous work and tables of historical progress.
The highest relative precision for $C$ so far, $10^{-7}$, was achieved
in 2010 by parallelizing a Monte Carlo method and running it on a PC
cluster~\cite{Hwang10}. See also Table~\ref{tab:refvals}.

The problem of determining $C$ can be modeled as an integral equation
much in the same way as the problem of determining $\alpha(z)$ and we
omit details. If one solves
\begin{equation}
\left(I+K+Q\right)\rho(r)=1\,,
\label{eq:sys2}
\end{equation}
where $K$ is as in~(\ref{eq:inteq2}) and $Q$ is the surface integral
operator
\begin{equation}
Q\rho=\int_S\rho(r)\,{\rm d}\sigma_r\,,
\end{equation}
then the (normalized) capacitance can be evaluated as
\begin{equation}
C=\frac{1}{4\pi U(r)}\,,\quad r\in V\,,
\end{equation}
where
\begin{equation}
U(r)=\int_S G(r,r')\rho(r')\,{\rm d}\sigma_{r'}\,.
\label{eq:rep2}
\end{equation}

\section{Strategies for computing $\alpha^+(x)$}
\label{sec:strat}

The difficulties with computing $\alpha^+(x)$ when $S$ has edges and
corners relate to issues of stability and resolution. We know from
Theorem~\ref{thm:Ublowup} that the electrostatic problem does not have a
finite energy solution $U(r)$ and that~(\ref{eq:inteq1}) does not have
a solution $\rho(r)\in H^{-1/2}(S)$ for $z\in\sigma_{\mu}$. Close to
$\sigma_{\mu}$, stability problems can be expected. A computational
mesh needs to be extremely refined (locally) in order to resolve
$U(r)$ or $\rho(r)$ and the solver must have the capability of dealing
with strongly singular solutions.

One strategy for alleviating these problems is to round edges and
corners so that $S$ becomes smooth. Then $U(r)$ has finite energy and
$\rho(r)\in L^2(S)$, except for at $z=z_i$, and $\alpha(z)$
assumes the form~(\ref{eq:lambda4b}). Some distance above the real
axis, corresponding to bigger losses, $\alpha(z)$ may resemble
$\alpha^+(x)$ and can be evaluated using commercial software. This is
essentially the approach in~\cite{Kettu08,Walle08,Zhang11}, where
finite element solvers are chosen.

It is, however, possible and advantageous to take the limit
$z\to\sigma_{\mu}$ in $\alpha(z)$ numerically while letting $S$ retain
its sharp shape. The rounding of edges and corners, while smoothing
solutions, introduces new length-scales which is an unnecessary
complication. In fact, there has been an intense activity in the area
of constructing numerical algorithms for solving integral equations on
non-smooth curves in recent
years~\cite{Brem12a,Brem12b,Brem10b,Brun09,Hels09JCP,Hels11SISC,Hels08,Hels09IJSS}.
See~\cite[Section~1.3]{Hels11JCP} for an overview and a comparison of
various approaches. Sharp corners and other boundary singularities can
be treated extremely efficiently using fast direct and fully automatic
solvers and by taking advantage of asymptotic self-similarity. An
algorithm to this effect for a quantity analogous to $\alpha(z)$ for
squares in a periodic setting is assembled and tested thoroughly
in~\cite{Hels11JCP}. This paper continues along the lines of that
work.

\section{Algorithm for the square}
\label{sec:two}

This section is a summary of results from
Refs.~\cite{Hels09JCP,Hels11JCP} applied to the solution
of~(\ref{eq:inteq2}) for $V$ being a square. We construct two meshes
on $S$ -- a coarse mesh and a fine mesh. The coarse mesh has 16
quadrature panels. The fine mesh is constructed from the coarse mesh
by $n_{\rm sub}$ times subdividing the panels closest to each corner
vertex $s_k$ in a direction towards the vertex. See
Fig.~\ref{fig:sqmesh}.

\begin{figure}[t]
\centering 
\includegraphics[height=39mm]{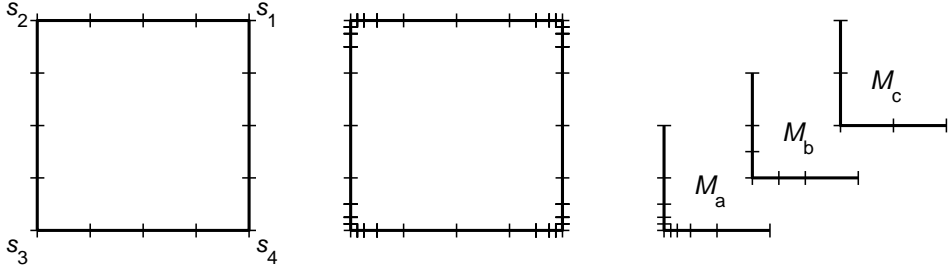}
\caption{Meshes on the boundary $S$ of a square. Left: the coarse mesh
  and the corner vertices $s_k$. Middle: a mesh that is refined
  $n_{\rm sub}=3$ times. Right: local meshes ${\cal M}_{\rm a}$,
  ${\cal M}_{\rm b}$ and ${\cal M}_{\rm c}$ centered around the corner
  vertex $s_3$.}
 \label{fig:sqmesh}
\end{figure}

\subsection{Preconditioning and discretization}
\label{sec:prec}

Let $S_k^{\star}$ denote a segment of $S$ covering the four panels on
the coarse mesh that lie closest to the corner vertex $s_k$ -- two
panels on each side of $s_k$. The $S_k^{\star}$ are disjoint and their
union is $S$. Let $K(r,r')$ denote the kernel of $K$
in~(\ref{eq:inteq2}).  Split $K(r,r')$ into two functions
\begin{equation}
K(r,r')=K^{\star}(r,r')+K^{\circ}(r,r')\,,
\label{eq:split}
\end{equation}
where $K^{\star}(r,r')$ is zero except for when $r$ and $r'$
simultaneously lie on the same $S_k^{\star}$. In this latter case
$K^{\circ}(r,r')$ is zero.

The kernel split~(\ref{eq:split}) corresponds to an operator split
$K=K^{\star}+K^{\circ}$ where $K^{\circ}$ is a compact operator.
Discretization of~(\ref{eq:inteq2}), using a Nystr{\"o}m method based
on composite 16-point Gauss--Legendre quadrature and a coarse or a
fine mesh on $S$, leads to an equation of the form
\begin{equation}
\left({\bf I}+\lambda{\bf K}^{\star}
                       +\lambda{\bf K}^{\circ}\right)
\boldsymbol{\rho}=\lambda{\bf g}\,,
\label{eq:gen1}
\end{equation}
where ${\bf I}$, ${\bf K}^{\star}$, and ${\bf K}^{\circ}$ are square
matrices and $\boldsymbol{\rho}$ and ${\bf g}$ are columns vectors.
The matrix ${\bf K}^{\star}$ assumes a block-diagonal structure since
the $S_k^{\star}$ are disjoint. This will be important in what
follows.

The change of variables
\begin{equation}
\rho(r)=\left(I+\lambda K^{\star}\right)^{-1}\tilde{\rho}(r)
\end{equation}
makes~(\ref{eq:gen1}) read
\begin{equation}
\left({\bf I}+\lambda{\bf K}^{\circ}
\left({\bf I}+\lambda{\bf K}^{\star}\right)^{-1}\right)
\tilde{\boldsymbol{\rho}}=\lambda{\bf g}\,.
\label{eq:gen2}
\end{equation}
This right preconditioned equation corresponds to the discretization
of a Fredholm second kind equation with a composed compact operator.
The solution $\tilde{\boldsymbol{\rho}}$ is the discretization of a
function that is piecewise smooth and can be resolved by piecewise
polynomials.

From now on we let subscripts {\footnotesize fin} and {\footnotesize
  coa} indicate what type of mesh is used for the discretization. The
collection of discretization points on a mesh is called a {\it grid}.
The number $n_{\rm sub}$ is assumed to be high enough so that
$\boldsymbol{\rho}_{\rm fin}$ resolves $\rho(r)$ to the precision
sought in our computations.

\subsection{Compression}
\label{sec:comp}

The following decomposed low-rank approximation of the discretization
of $K^{\circ}$ on the fine mesh holds to very high accuracy:
\begin{equation}
{\bf K}_{\rm fin}^{\circ}\approx
{\bf P}{\bf K}_{\rm coa}^{\circ}{\bf P}^T_W\,.
\label{eq:approx1}
\end{equation}
Here ${\bf K}_{\rm fin}^{\circ}$ is a $(256+96n_{\rm sub})\times
(256+96n_{\rm sub})$ matrix, ${\bf K}_{\rm coa}^{\circ}$ is a
$256\times 256$ matrix, ${\bf P}$ is a prolongation matrix from the
coarse grid to the fine grid and
\begin{equation}
{\bf P}_W={\bf W}_{\rm fin}{\bf P}{\bf W}_{\rm coa}^{-1}\,,
\label{eq:prolongW}
\end{equation}
where ${\bf W}$ is a diagonal matrix containing the quadrature weights
of the discretization, see~\cite[Section~5]{Hels09JCP}. Superscript
{\footnotesize $T$} denotes the transpose.

The relation~(\ref{eq:approx1}) has powerful consequences for
computational efficiency in the context of solving~(\ref{eq:gen2}). As
we soon shall see, it allows us to compress that equation and obtain
the accuracy offered by the fine grid while working chiefly on the
coarse grid. Only $\left({\bf I}+\lambda{\bf K}^{\star}\right)^{-1}$
needs the fine grid for resolution. We introduce the compressed
quadrature-weighted inverse
\begin{equation}
{\bf R}={\bf P}^T_W
\left({\bf I}_{\rm fin}+\lambda{\bf K}_{\rm fin}^{\star}\right)^{-1}{\bf P}\,.
\label{eq:R0}
\end{equation}
With~(\ref{eq:R0}), the discretization of~(\ref{eq:inteq2}) assumes
the final form
\begin{equation}
\left({\bf I}_{\rm coa}+\lambda{\bf K}_{\rm coa}^{\circ}{\bf R}\right)
\tilde{\boldsymbol{\rho}}_{\rm coa}=\lambda{\bf g}_{\rm coa}\,,
\label{eq:final}
\end{equation}
where all matrices are $256\times 256$. For later reference we introduce
\begin{equation}
\hat{\boldsymbol{\rho}}_{\rm coa}={\bf R}\tilde{\boldsymbol{\rho}}_{\rm coa}
\end{equation}
as a {\it weight-corrected} density~\cite[Section~5]{Hels09IJSS}.

\subsection{Recursion}
\label{sec:recur}

The compressed inverse ${\bf R}$ has a block diagonal structure with
four identical $64\times 64$ blocks ${\bf R}_k$ associated with the
vertices $s_k$. The construction of a block ${\bf R}_k$ from the
definition~(\ref{eq:R0}) is costly when $n_{\rm sub}$ is large.
Fortunately, this construction can be greatly sped up via a recursion.
In general situations this recursion uses grids on hierarchies of
local meshes, see~\cite[Section~6]{Hels09JCP}
and~\cite[Section~5]{Hels11SISC}. For wedge-like corners, thanks to
scale invariance of $K(r,r')$, only two local meshes ${\cal M}_{\rm
  b}$ and ${\cal M}_{\rm c}$ are needed, see right image of
Fig.~\ref{fig:sqmesh}. The recursion for block ${\bf R}_k$ assumes the
form of a simple fixed-point iteration
\begin{equation}
{\bf R}_{ik}={\bf P}^T_{W\rm{bc}}
\left(
\mathbb{F}\{{\bf R}_{(i-1)k}^{-1}\}
+{\bf I}_{\rm b}^{\circ}+\lambda{\bf K}_{{\rm b}k}^{\circ}
\right)^{-1}{\bf P}_{\rm{bc}}\,,\quad i=1,\ldots,n_{\rm sub}\,,
\label{eq:genrec}
\end{equation}
where ${\bf R}_{ik}={\bf R}_{k}$ for $i=n_{\rm sub}$. The
quadrature weighted and unweighted prolongation matrices ${\bf
  P}_{W\rm{bc}}$ and ${\bf P}_{\rm{bc}}$ act from a grid on a local
mesh ${\cal M}_{\rm c}$ to a grid on a local mesh ${\cal M}_{\rm b}$.
The superscript $\circ$ in~(\ref{eq:genrec}) has a similar meaning as
in~(\ref{eq:split}) and the operator $\mathbb{F}\{\cdot\}$ expands a
matrix by zero-padding, see~\cite[Section~6]{Hels11JCP}.

The derivation of~(\ref{eq:genrec}) relies on a low-rank approximation
similar to~(\ref{eq:approx1})
\begin{equation}
{\bf K}_{\rm a}^{\circ}\approx
{\bf P}_{\rm ab}{\bf K}_{\rm b}^{\circ}{\bf P}^T_{W{\rm ab}}\,,
\label{eq:approx2}
\end{equation}
where ${\bf K}_{\rm a}$ is a discretization of $K$ on a multiply
refined local mesh ${\cal M}_{\rm a}$. See~\cite[Section~7]{Hels08}
for details. Conceptually, one could think of~(\ref{eq:genrec}) as a
process on a multiply refined local mesh, going outwards from the
vertex, where step $i$ inverts and compresses contributions to ${\bf
  R}_k$ involving the outermost panels on level $i$.

The number $n_{\rm sub}$ needed for resolution of ${\bf R}_k$ may grow
without bounds as $z$ approaches $\sigma_{{\mu}{\rm sq}}$ (in infinite
precision arithmetic). In order to accelerate the recursion we
activate a combination of numerical homotopy and Newton's method when
deemed worthwhile, see~\cite[Section~6]{Hels11JCP}. The Newton
iterations are continued until the relative update in ${\bf R}_{ik}$
is smaller than $100\epsilon_{\rm mach}$ or a maximum number of 20
iterations is reached, which roughly corresponds to a local mesh
${\cal M}_{\rm a}$ that is refined $n_{\rm sub}\approx 2^{20}\approx
10^6$ times.

\subsection{Solution, post-processing and interpretations}
\label{sec:interp}

Once the $256\times 256$ linear system~(\ref{eq:final}) is solved for
$\tilde{\boldsymbol{\rho}}_{\rm coa}$, various quantities of interest
can be computed. For example, the polarizability~(\ref{eq:alpha2})
becomes
\begin{equation}
\alpha(z)={\bf h}_{\rm coa}^T
{\bf W}_{\rm coa}{\bf R}\tilde{\boldsymbol{\rho}}_{\rm coa}\,,
\label{eq:alpha3}
\end{equation}
where ${\bf h}$ is the discretization of $h(r)$. Results produced in
this way are extremely accurate and fully confirm 24 of the entries
for $\alpha(z)$ with $|z|\ge 1$ in Table~1 of~\cite{Milt81}. The
remaining 5 entries differ in the last digit. Section~8
of~\cite{Hels11JCP} gives error estimates for results produced in
periodic settings.

The original density ${\boldsymbol{\rho}}_{\rm fin}$ can be
reconstructed from $\tilde{\boldsymbol{\rho}}_{\rm coa}$ by, in a
sense, running the recursion~(\ref{eq:genrec}) backwards (inwards on a
multiply refined local mesh). If this process is interrupted part-way,
one is left with a mix of discrete values of the original density
$\boldsymbol{\rho}$ (on outer panels) and quantities which can be
easily converted into discrete values of a weight-corrected density
$\hat{\boldsymbol{\rho}}$ (on the innermost panels). Details of the
process are given in~\cite[Section~7]{Hels09JCP}. Here it suffices to
observe that there exists a rectangular matrix ${\bf Y}$, say, whose
action on $\tilde{\boldsymbol{\rho}}_{\rm coa}$ produces entries of
$\boldsymbol{\rho}_{\rm fin}$.

Let ${\bf Q}$ denote a restriction matrix from the fine grid to the
coarse grid. Then
\begin{equation}
\boldsymbol{\rho}_{\rm coa}
={\bf Q}{\bf Y}\tilde{\boldsymbol{\rho}}_{\rm coa}
\end{equation}
and 
\begin{equation}
{\bf K}_{\rm coa}^{\circ}{\bf R}
\tilde{\boldsymbol{\rho}}_{\rm coa}=
{\bf K}_{\rm coa}^{\circ}{\bf R}\left({\bf Q}{\bf Y}\right)^{-1}
{\boldsymbol{\rho}}_{\rm coa}\,.
\end{equation}
We see that the blocks of the block-diagonal matrix ${\bf R}\left({\bf
    Q}{\bf Y}\right)^{-1}$ have an interpretation as multiplicative
weight corrections needed if ${\bf K}_{\rm coa}^{\circ}$ is to act
accurately on $\boldsymbol{\rho}_{\rm coa}$. Other useful
interpretations of the matrices introduced include: The columns of
${\bf Y}$ are discrete basis functions for $\rho(r)$ on the fine grid;
The columns of ${\bf R}$ are discrete basis functions for $\rho(r)$ on
the coarse grid multiplied with quadrature weights; the rectangular
matrix ${\bf Y}\left({\bf Q}{\bf Y}\right)^{-1}$ maps
$\boldsymbol{\rho}_{\rm coa}$ to $\boldsymbol{\rho}_{\rm fin}$. These
observations will be used in what follows.

The machinery of Sections~\ref{sec:prec},~\ref{sec:comp}
and~\ref{sec:recur} is useful in several ways. The preconditioning
aspect of~(\ref{eq:final}) reduces numerical error and improves the
convergence of iterative solvers. The compression aspect
of~(\ref{eq:final}) saves degrees of freedom and makes the algorithm
fast and memory efficient. The recursion~(\ref{eq:genrec}) resolves
the singular nature of $\rho(r)$ close to corner vertices in an
automated fashion and provides efficient basis functions and
quadrature weights contained in the matrix ${\bf R}$. No asymptotic
analysis is required -- we simply use Gauss-Legendre quadrature on the
coarse mesh and on the local meshes ${\cal M}_{\rm b}$ and ${\cal
  M}_{\rm c}$ on which the recursion~(\ref{eq:genrec}) takes place.
Most, but not all, of these features can be retained as we step up
into three dimensions.

\begin{figure}[t]
\centering 
\includegraphics[height=36mm]{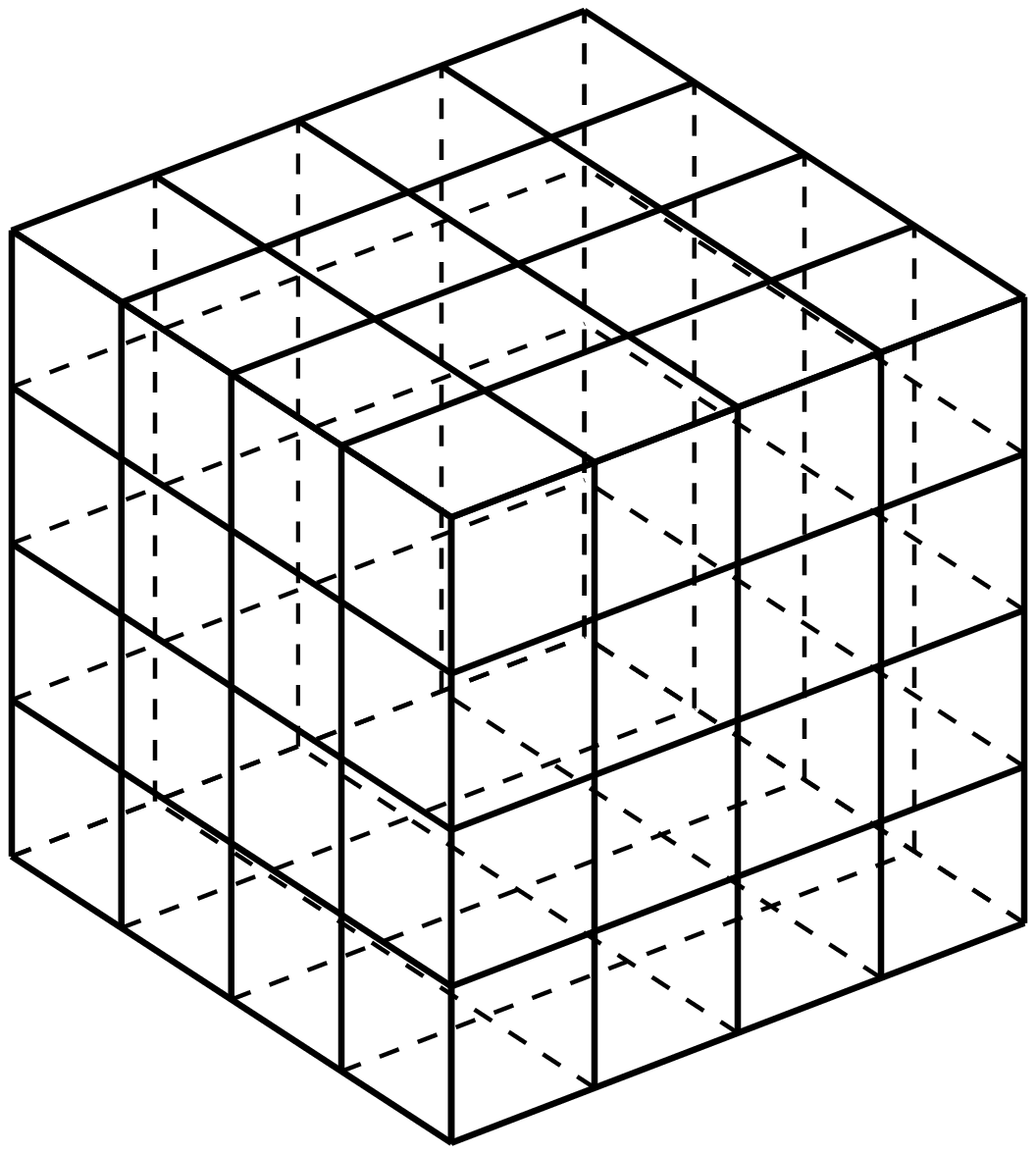}
\hspace{5mm}
\includegraphics[height=36mm]{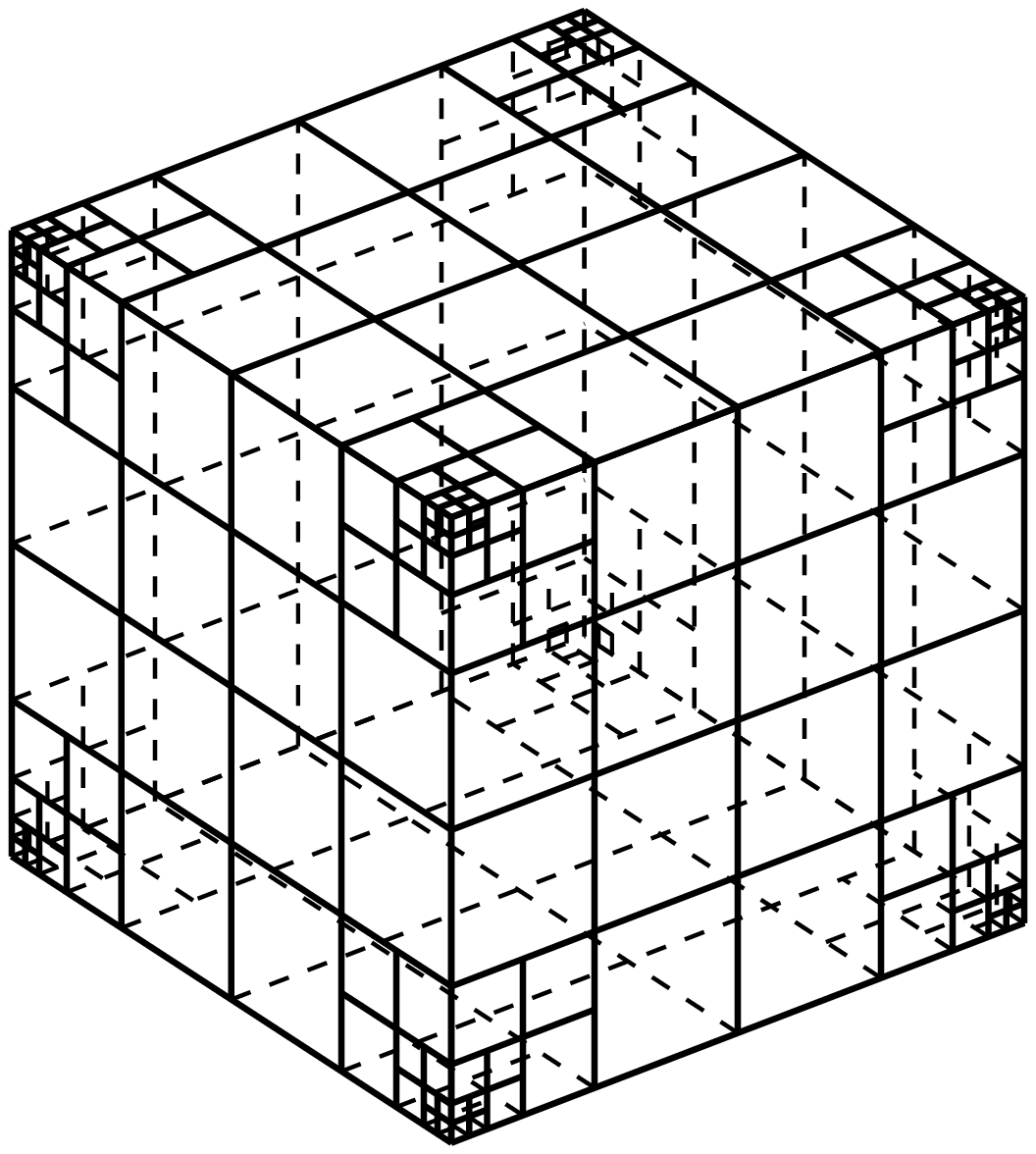}
\hspace{4mm}
\includegraphics[height=36mm]{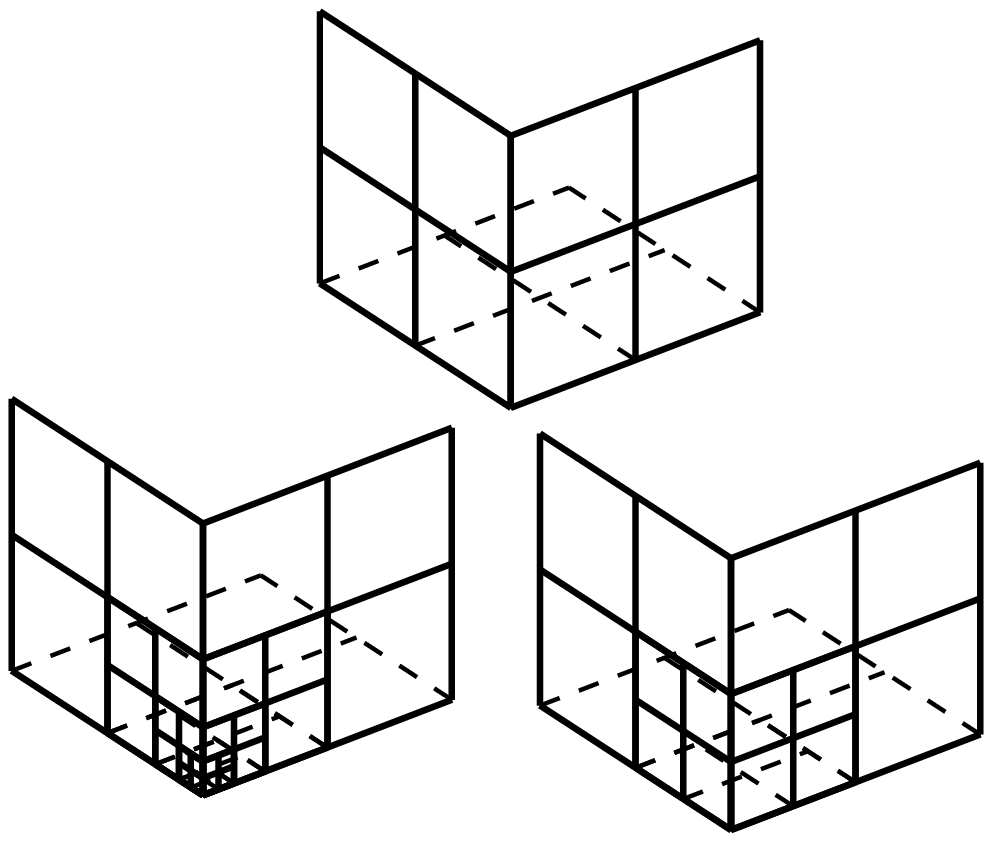}
\caption{Meshes on the surface $S$ of a cube. Left: the coarse mesh. 
  Middle: a mesh that is refined $n_{\rm sub}=3$ times. Right: local
  meshes ${\cal M}_{\rm a}$, ${\cal M}_{\rm b}$ and ${\cal M}_{\rm c}$
  centered around a corner vertex. Compare Fig.~\ref{fig:sqmesh}}
 \label{fig:cumesh}
\end{figure}

\section{Algorithm for the cube}
\label{sec:three}

Our algorithm for the cube mimics that of Section~\ref{sec:two}. A
problem, however, arises in the split~(\ref{eq:split}). Unlike the
square, the cube has both sharp edges and sharp corners and it is not
possible to identify suitably disjoint surface elements $S_k^{\star}$
that allow for an operator split $K=K^{\star}+K^{\circ}$ where
$K^{\circ}$ is a compact operator. Thus, one cannot construct a
block-diagonal matrix ${\bf R}$ which contains weighted basis
functions for $\rho(r)$ that simultaneously resolve the singularities
stemming from the edges and from the corners. We shall circumvent this
difficulty by focusing solely on the cube corners and construct the
coarse and the fine mesh with square quadrature panels according to
Fig.~\ref{fig:cumesh}, that is, in complete analogy with
Fig.~\ref{fig:sqmesh}. As to compensate for the lack of refinement
towards the edges, the discretization of $K^{\star}$ and $K^{\circ}$
will incorporate ready-made one-dimensional basis functions and
weights constructed for a square with the same $\lambda$ according to
Section~\ref{sec:interp}.

\subsection{Preconditioning and discretization}
\label{sec:prec2}

This section is a counterpart to Section~\ref{sec:prec}. Let the eight
corner vertices of the cube be denoted $s_k$ and introduce surface
element $S_k^{\star}$ that cover the 12 quadrature panels on the
coarse mesh that lie closest to $s_k$. Furthermore, let the smaller
surface elements $S_k^{\star\star}$ be such that they cover the three
panels on the coarse mesh that lie closest to $s_k$. The $S_k^{\star}$
are disjoint, their union is $S$, and the kernel
split~(\ref{eq:split}) results in an operator split where the part of
$K^{\circ}$ which accounts for interaction between points in
$S\setminus S_k^{\star}$ and $S_k^{\star\star}$ is compact. This
restricted compactness is sufficient for our purposes since mesh
refinement only takes place on the $S_k^{\star\star}$, see the middle
image of Fig.~\ref{fig:cumesh}.

We discretize~(\ref{eq:inteq2}) and make the change of variables that
leads to~(\ref{eq:gen2}). The discretization of $K$ proceeds in three
steps. First, we use the Nystr{\"o}m method applying tensor products
of $n_{\rm p}$-point Gauss--Legendre quadrature formulas on all
quadrature panels to get an initial ${\bf K}$. Then, for columns of
${\bf K}$ acting on discrete densities on panel pairs neighboring a
single cube edge, but not a corner, we correct the Gauss--Legendre
weights in the direction perpendicular to the edge. These corrections
are realized by multiplying submatrices of ${\bf K}$ with blocks of
${\bf R}_{\rm sq}\left({\bf Q}_{\rm sq}{\bf Y}_{\rm sq}\right)^{-1}$,
where subscript {\footnotesize sq} indicates the square. See
Section~\ref{sec:interp}. This is our {\it basic discretization}. The
resulting ${\bf K}$ acts accurately on $\boldsymbol{\rho}$ in
situations describing the convolution of $K(r,r')$ with $\rho(r')$
both for $r'$ away from edges and corners and for $r'$ close to an
edge but away from corners and $r$ away from $r'$.

Lastly, we change blocks of ${\bf K}$ describing interaction between
panel pairs neighboring each other on opposite sides of an edge. Such
interaction requires special attention due to the non-smooth nature of
$K(r,r')$ close to edges. We use interpolatory quadrature based on
polynomial basis functions in the direction parallel to the edge and
on the basis functions of ${\bf Y}_{\rm sq}$ in the direction
perpendicular to the edge. This gives our final ${\bf K}$. The total
number of discretization points on the coarse mesh is $n=96n_{\rm
  p}^2$. The fine mesh has $(96+72n_{\rm sub})n_{\rm p}^2$ points.

The choice of the columns of ${\bf Y}_{\rm sq}$ as discrete basis
functions for $\rho(r)$ in the direction perpendicular to edges should
be asymptotically correct, assuming that the singularities $\rho(r)$
are dominated by two-dimensional effects away from the corner. Still,
these basis functions are not optimal and they are responsible for the
slower rate of convergence that we shall see when
solving~(\ref{eq:final}) for the cube, compared to when
solving~(\ref{eq:final}) for the square.

\subsection{Compression, recursion and post-processing}

The compression of~(\ref{eq:gen2}) for the cube is analogous to that
for the square in Section~\ref{sec:comp}. The only difference, apart
from that various matrices have different sizes, is that ${\bf W}$,
which contains the quadrature weights of the basic discretization and
enters into the definition of prolongation matrix ${\bf P}_W$
of~(\ref{eq:prolongW}), is no longer diagonal. The blocks of ${\bf
  R}_{\rm sq}\left({\bf Q}_{\rm sq}{\bf Y}_{\rm sq}\right)^{-1}$, used
as multiplicative weight corrections, are generally full matrices.

The recursion for ${\bf R}$ of the cube, from now on denoted ${\bf
  R}_{\rm cu}$, follows Section~\ref{sec:recur} exactly. Some
$\lambda$ allow for a rapid convergence in~(\ref{eq:genrec}). Other
$\lambda$, corresponding to $z$ close to parts of $\sigma_{{\mu}{\rm
    cu}}$, require that we resort to Newton's method and numerical
homotopy. Note that recursions are carried out twice in the scheme:
both for ${\bf R}_{\rm sq}$, needed for the discretization, and for
${\bf R}_{\rm cu}$.

The polarizability $\alpha(z)$ of the cube can be computed
from~(\ref{eq:alpha3}) once the $96n_{\rm p}^2\times 96n_{\rm p}^2$
system~(\ref{eq:final}) is solved. The matrix ${\bf W}_{\rm coa}$
in~(\ref{eq:alpha3}) is now weight-corrected as described in the first
paragraph of this section. We use the GMRES iterative solver
for~(\ref{eq:final}) and make some use of symmetry in order to reduce
memory requirements.

\section{From circle to square}
\label{sec:numsquare}

In a series of numerical experiments we now study the spectrum of $K$
and the polarizability $\alpha(z)$ for a surface $S$ that is gradually
transformed from smooth to non-smooth. Such a study is of interest for
several reasons. Due to the difficulties associated with solving
electrostatic problems on non-smooth domains, see
Section~\ref{sec:strat}, it is common to round sharp boundary features
prior to discretization~\cite{Kettu08,Walle08,Zhang11}. Numerical
effects, similar to those caused by rounding, could also result from
insufficient resolution~\cite{Perr10}. Furthermore, no edge or corner
in real world physics is infinitely sharp and the degree of edge
smoothness can be critical in the design of, for example,
nanoantennas~\cite{Grill11}. Finally, the experiments illustrate the
theory overview of Section~\ref{sec:theory} in a setting which allows
for high accuracy. All experiments are executed on a workstation
equipped with an IntelCore2 Duo E8400 CPU at 3.00 GHz and 4 GB of
memory.

\begin{figure}[t]
\centering
 \includegraphics[height=51mm]{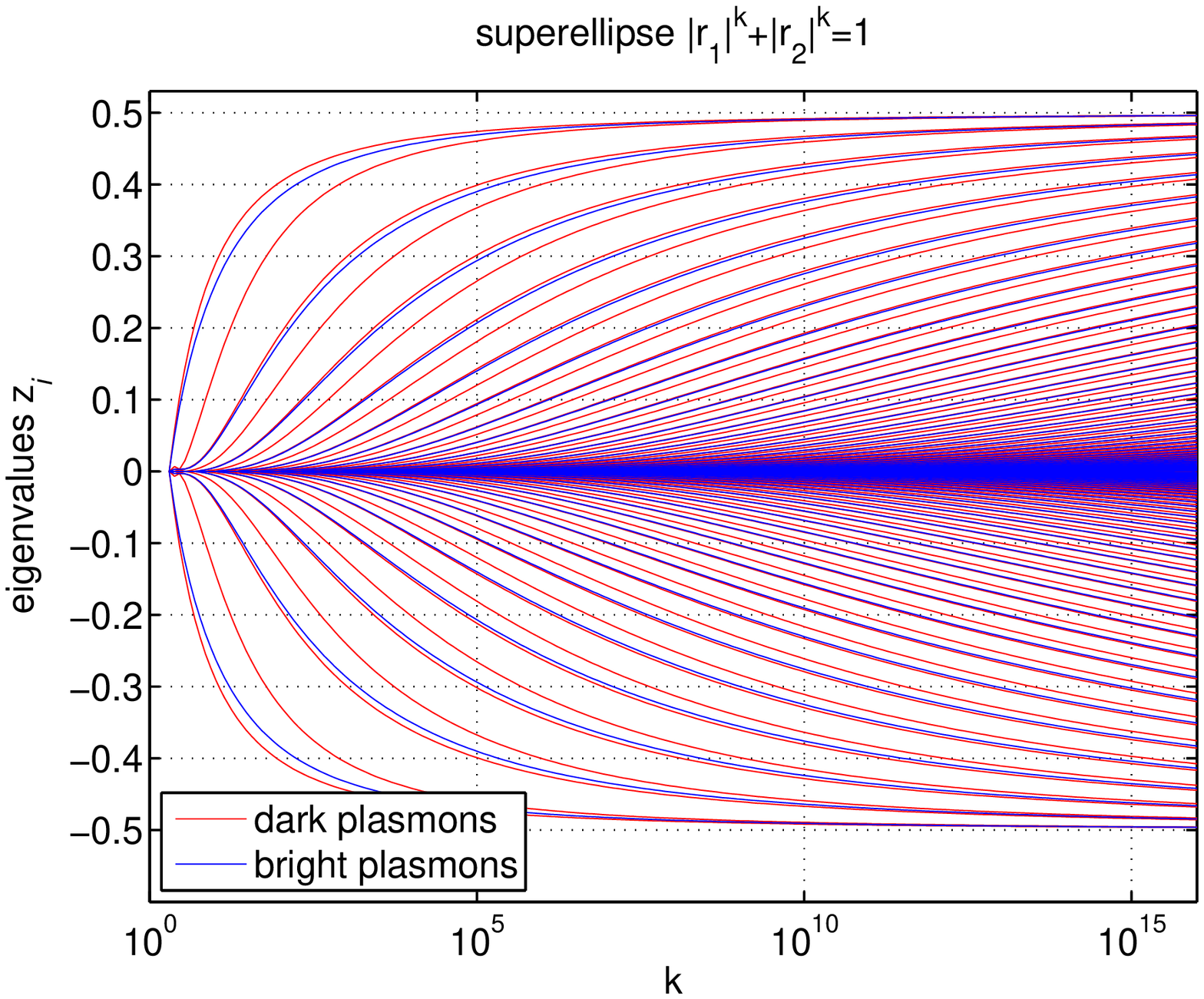}
 \includegraphics[height=51mm]{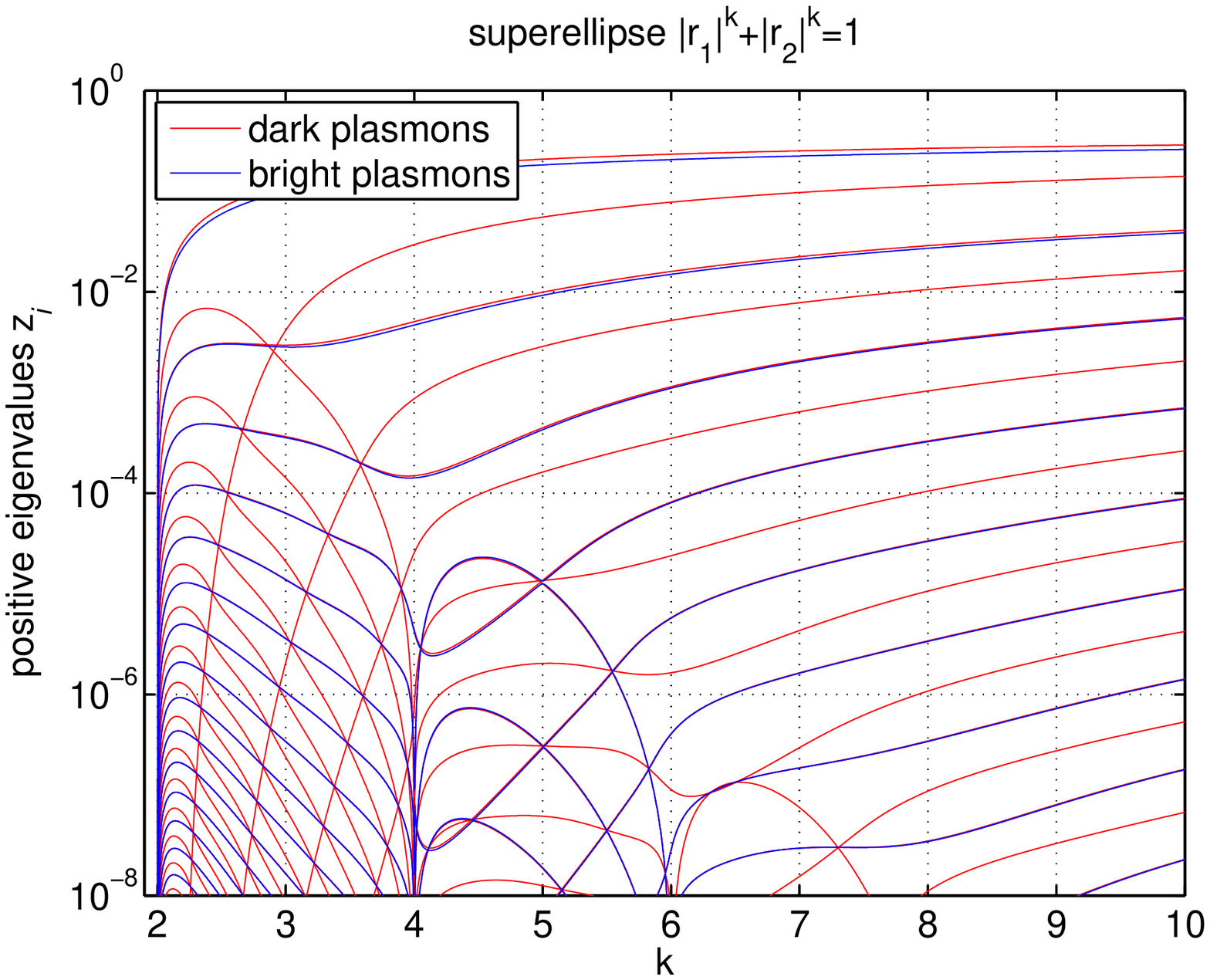}
 \caption{The eigenvalues $z_i$ of $K$ (locations of plasmons) for the
   superellipse varies with $k$. Left: $z_i$ with $2\le k\le 10^{16}$.
   The eigenvalue at $-1$, a dark plasmon, is omitted. Right: zoom of
   positive eigenvalues with $2\le k\le 10$.}
\label{fig:flowb}
\end{figure}

Similar to Klimov~\cite{Klimo08} we let $S$ be the superellipse
\begin{equation}
|r_1|^k+|r_2|^k=1\,,
\label{eq:super}
\end{equation}
which for $k=2$ is a circle and for $k\to\infty$ approaches a square.
We first compute eigenvalues of $K$ using a discretization based on
composite 16-point Gauss--Legendre quadrature and adaptive mesh
refinement. Particular care is taken in the parameterization of $S$ as
to allow for resolution at boundary portions of high curvature.  The
accuracy in these computations varies with $k$. A rough estimate is a
relative error of $\log_{10}(k)\cdot\epsilon_{\rm mach}$.

Eigenvalues and the nature of their corresponding plasmons are shown
in Fig.~\ref{fig:flowb}. The only bright plasmon at $k=2$ is a dipole.
When $k>2$, the bright plasmons have potential fields that are a mix
of modes (dipoles, octupoles, etc.). The left image of
Fig.~\ref{fig:flowb} shows that $K$ of the superellipse at
$k=10^{16}$, which is close to a square in double precision
arithmetic, has a spectrum that does not look continuous to the eye.
The right image of Fig.~\ref{fig:flowb} zooms in on the spectrum at
low $k$. Klimov~\cite{Klimo08}, in an analogous study for a
superellipsoid with $2\le k\le 6$, observes a phase-transition at
$k=2.5$ and a critical point at $k\approx 3$.

\begin{figure}[t]
  \centering 
  \includegraphics[height=51mm]{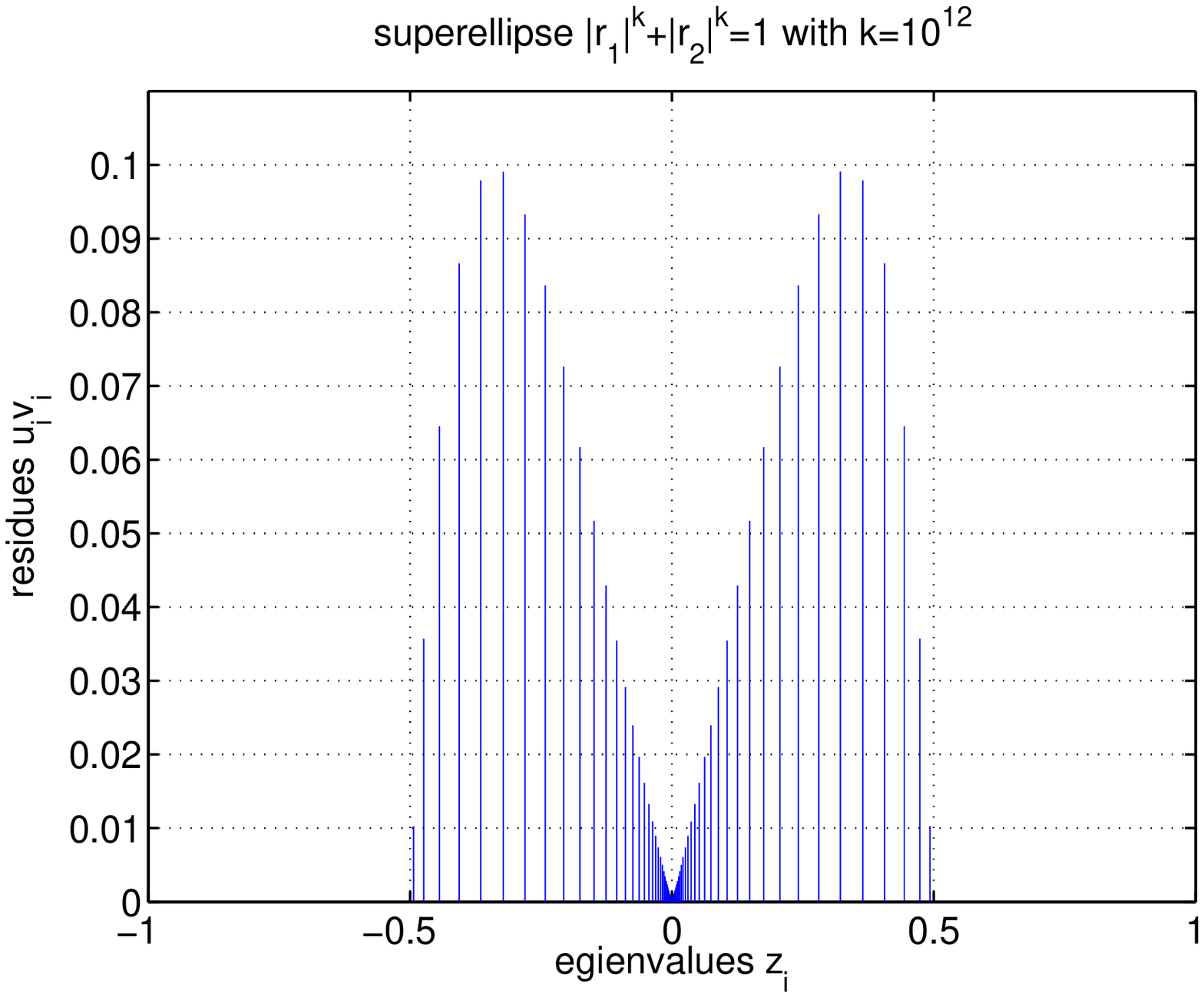}
  \includegraphics[height=51mm]{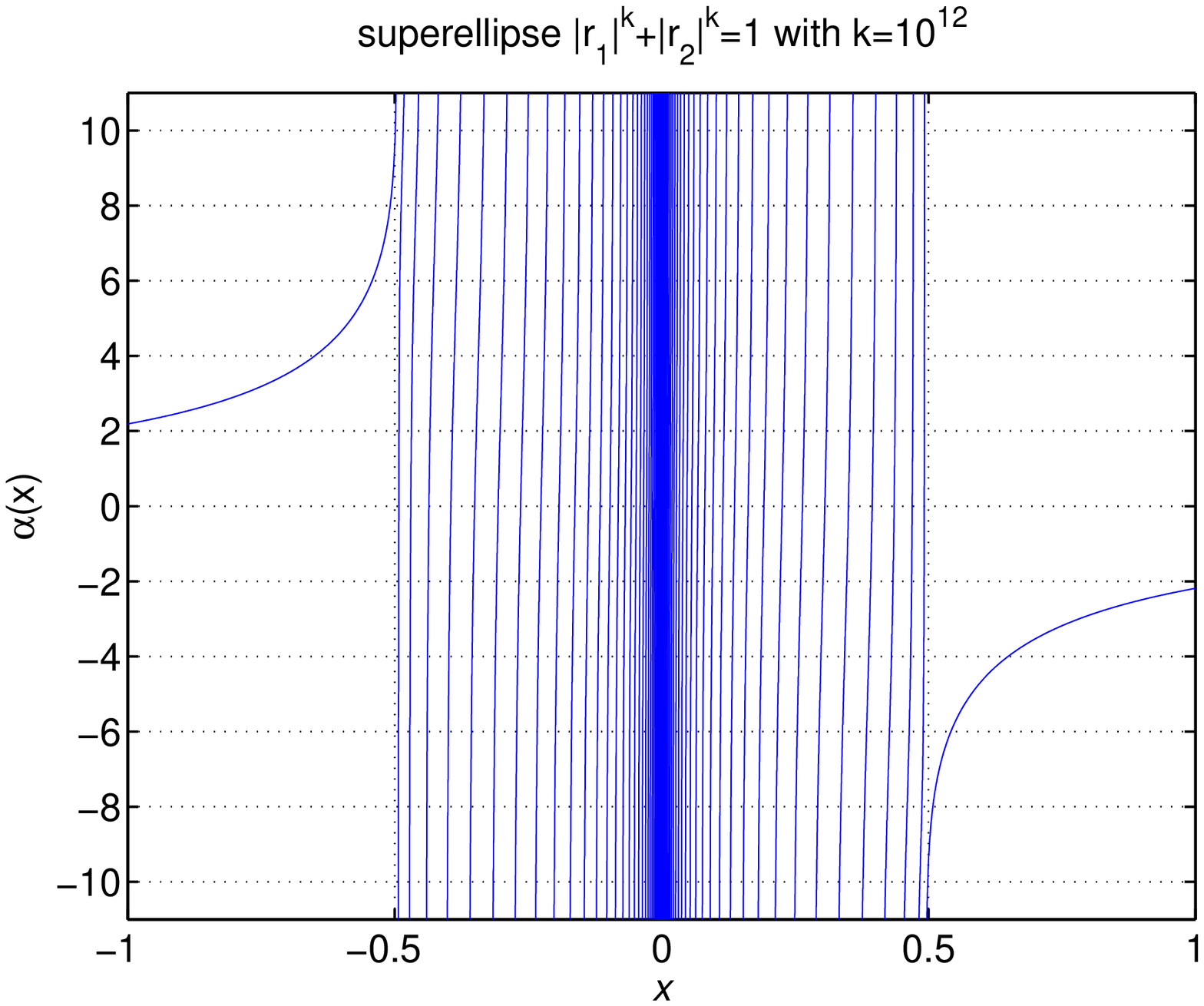}
  \includegraphics[height=51mm]{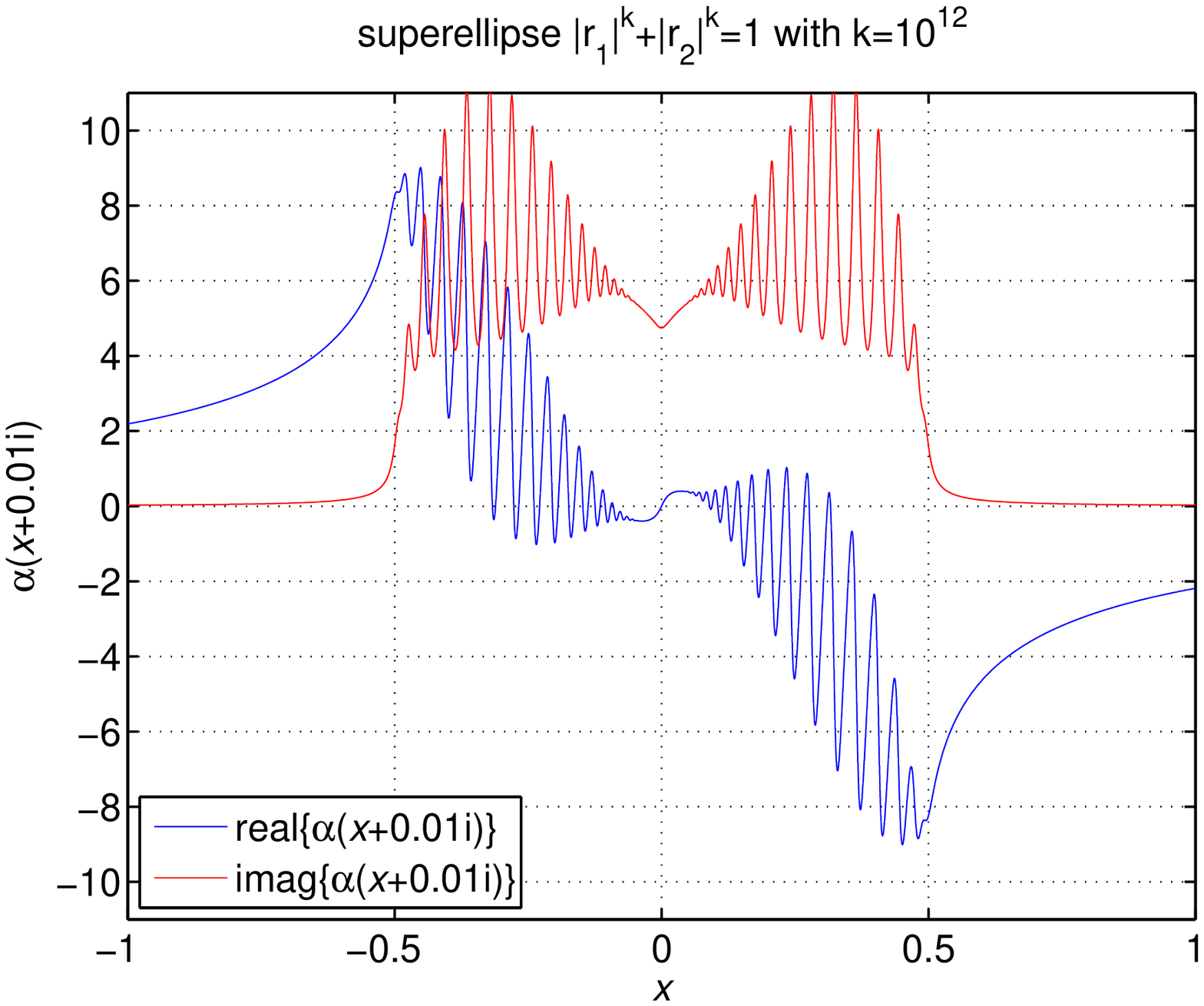}
  \includegraphics[height=51mm]{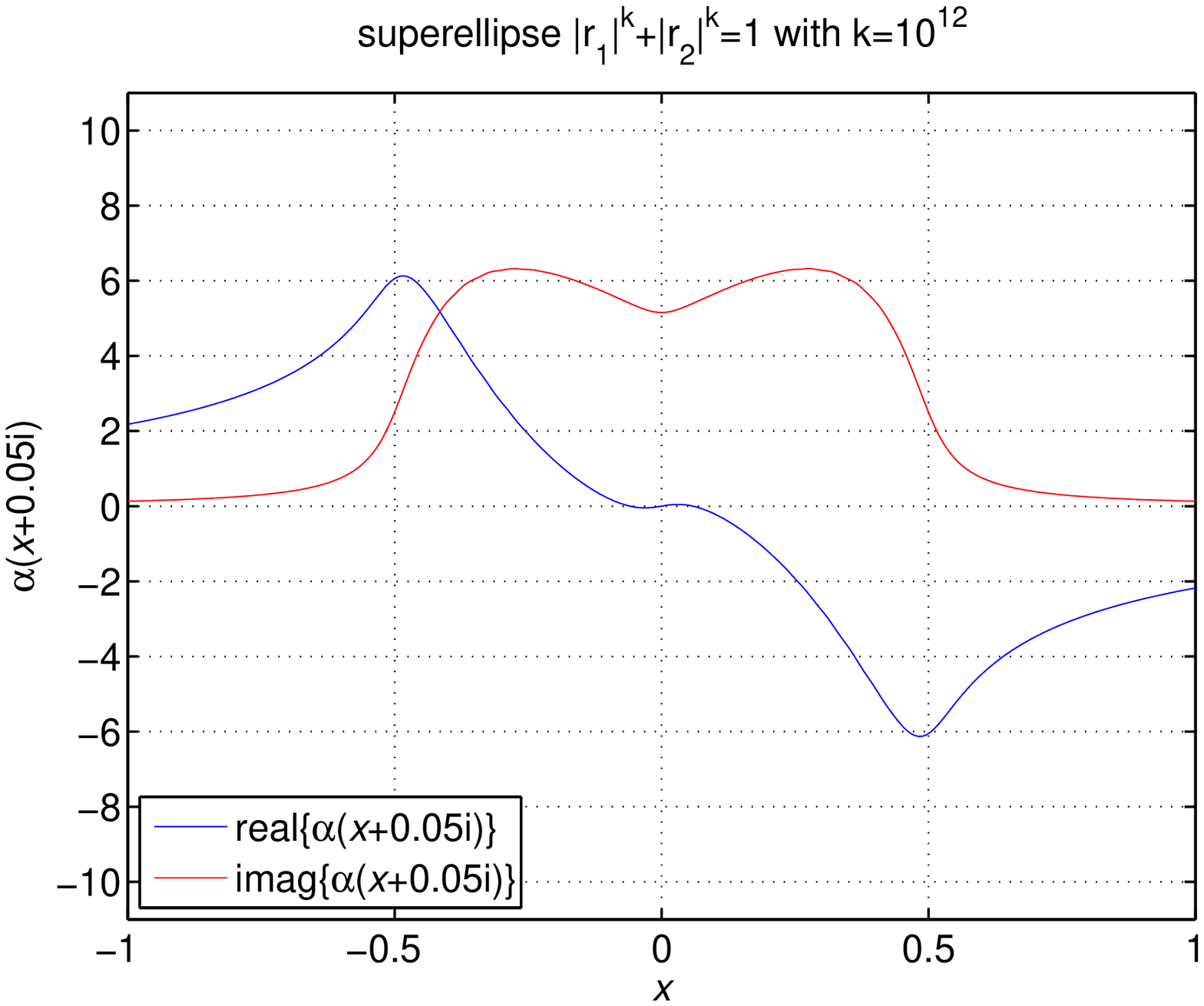}
  \caption{The number of eigenvalues $z_i$ with comparatively large 
    residues $u_iv_i$ for the superellipse (dominant plasmon modes)
    depends on $k$, and $\alpha(z)$ of~(\ref{eq:lambda4b}) decays away
    from $z\in[-0.5,0.5]$. Here $k=10^{12}$. Upper left: the 208
    largest residues. Upper right: $\alpha(x)$ with $x\in[-1,1]$.
    Lower left: $\alpha(x+0.01{\rm i})$. Lower right:
    $\alpha(x+0.05{\rm i})$. Compare Figs.~\ref{fig:flowb}
    and~\ref{fig:polabsq}.}
\label{fig:flowc}
\end{figure}

The number of eigenvalues $z_i$ with comparatively large residues
$u_iv_i$ in $\alpha(z)$ of~(\ref{eq:lambda4b}) depends on $k$.
Fig.~\ref{fig:flowc} shows residues and polarizabilities $\alpha(z)$
for $k=10^{12}$. Fig.~\ref{fig:flowc} also shows how $\alpha(z)$
approaches a slowly varying function as $z$ migrates from
$\overline{\sigma_{{\mu}{\rm sq}}}=[-0.5,0.5]$.

\begin{figure}[htb!]
  \centering \includegraphics[height=51mm]{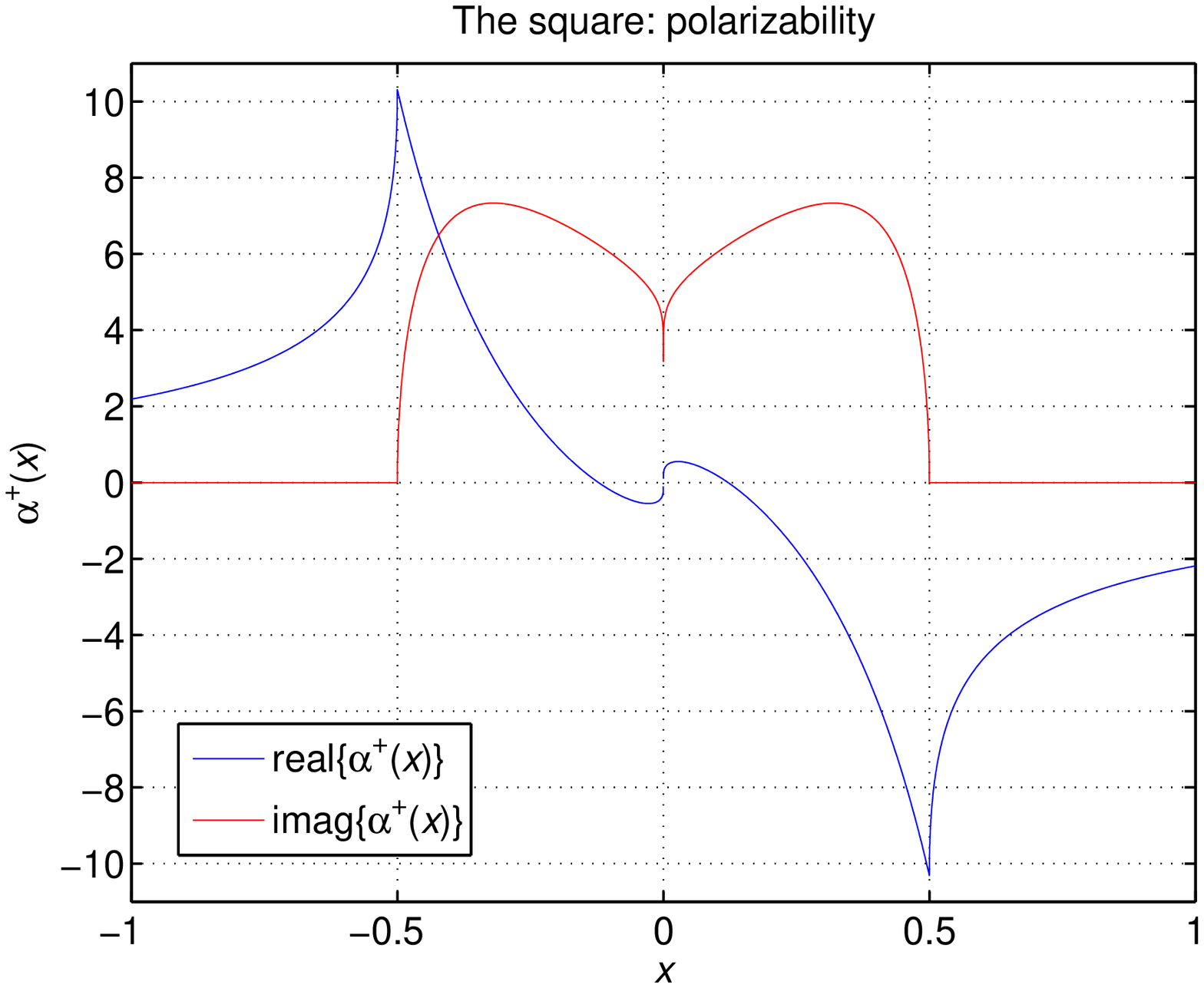}
  \includegraphics[height=51mm]{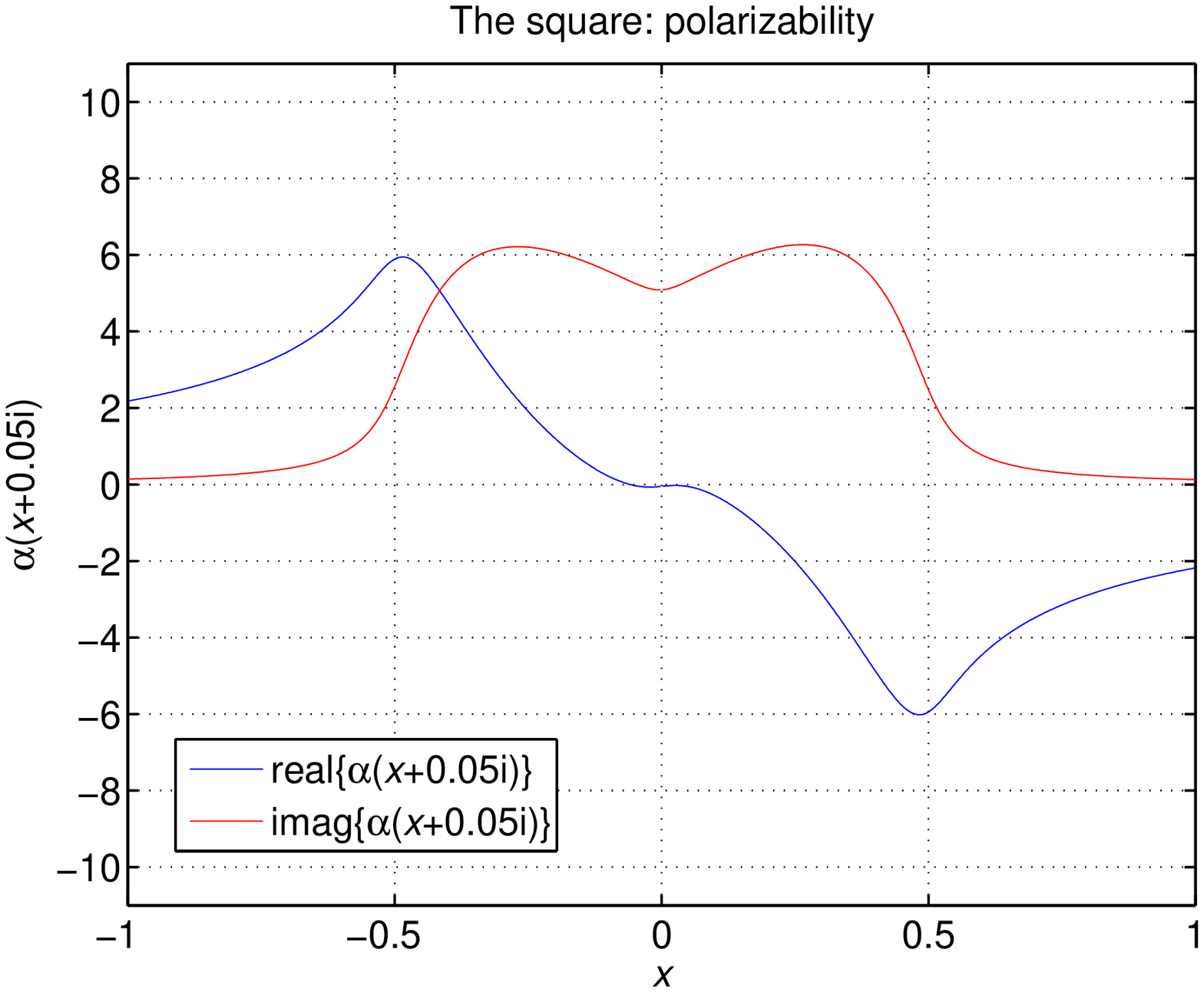}
\caption{Polarizability of a square. Left: $\alpha^+(x)$. The curves 
  are supported by 16492 data points whose relative accuracies range
  from machine precision to five digits. No convergent results were
  obtained within a distance of $10^{-7}$ from $x=0$. The values of
  $\Re\{\alpha^+(x)\}$ at $x=\mp 0.5$ are $\pm 10.3121215292$. The sum
  rule~(\ref{eq:sum1}), evaluated via~(\ref{eq:herglotz}) using a
  composite trapezoidal rule, holds to a relative precision of
  $10^{-6}$. Right: $\alpha(x+0.05{\rm i})$.}
\label{fig:polabsq}
\end{figure}

Fig.~\ref{fig:polabsq} compares $\alpha^+(x)$ with $\alpha(x+0.05{\rm
  i})$ for a square. The algorithm of Section~\ref{sec:two} is used.
Each data point takes only a fraction of a second to compute and is
accurate almost to machine precision except for $x$ very close to
$\{-0.5,0,0.5\}$. For example, the relative difference between the
computed value of $\alpha^+(1)$ and its known value
$-\Gamma(\frac{1}{4})^4/8\pi^2$, see~\cite{Thor92}, is $2\cdot
10^{-16}$. The left image of Fig.~\ref{fig:polabsq} shows that the
square has no bright plasmons (no poles in $\alpha^+(x)$) and stands
in forceful contrast the top right image of Fig.~\ref{fig:flowc},
which exhibits a myriad of plasmons for the superellipse at
$k=10^{12}$.

It is interesting to compare the right image of Fig.~\ref{fig:polabsq}
with the lower right image of Fig.~\ref{fig:flowc}. Already at a
distance of $0.05$ away from the real axis, $\alpha(z)$ of the square
and $\alpha(z)$ of the superellipse at $k=10^{12}$ are similar, as to
be expected in view of Theorem~\ref{thm:wstarconv}. The convergence as
$k\to\infty$ is uniform in $z$ in compact sets away from $[-1,1]$.
The accuracy achieved and the time required to evaluate $\alpha(z)$,
however, are very different. While the accuracy in $\alpha(z)$ of the
superellipse, computed via~(\ref{eq:lambda4b}), is perhaps four digits
and involves the eigenvalue decomposition of a $5376\times 5376$
matrix, the accuracy in $\alpha(z)$ of the square, computed
via~(\ref{eq:final}), is much higher and involves only computations
with matrices of size $256\times 256$. We conclude that even when
corners need not be strictly sharp from a physical viewpoint, it pays
off to keep them sharp from a numerical viewpoint.

\section{The cube}
\label{sec:numcube}

This section presents numerical results for the polarizability
$\alpha^+(x)$ and capacitance $C$ of the unit cube produced by the
algorithm of Section~\ref{sec:three}. While computing $C$ has a long
history in computational
electromagnetics~\cite{Bontz11,Hwang10,Mukh09}, computing
$\alpha^+(x)$ is less explored territory~\cite{Sihvo04} that only
recently, with the rapid growth of the field of nanotechnology, has
become fashionable. For example, $\sigma_{{\mu}{\rm cu}}=\left\{x:
  \mu'(x)>0\right\}$ seems to be largely unknown.
Fuchs~\cite{Fuchs75} and Langbein~\cite{Lang76}, over thirty years
ago, found six or eleven approximate eigenvalues for $K$ of the cube
(``major/dipole absorption peaks'') in the intervals $[-0.573,0.408]$
and $[-0.586,0.274]$, respectively and which supposedly account for
around 95\% of the sum~(\ref{eq:resid1}). In view of the findings of
the present paper, so far, one could suspect that these peaks are an
artifact of insufficient resolution or unintended rounding of edges.
Be as it may, this pioneering and highly cited work is now of interest
in nanoplasmonics where the coupling of plasmon modes in metallic
nanostructures such as nanocube dimers is important and
interpretations seem to rely on these
peaks~\cite{Grill11,Klimo08,Rupp96,Zhang11}.

\begin{figure}[t!]
 \centering 
 \includegraphics[height=88mm]{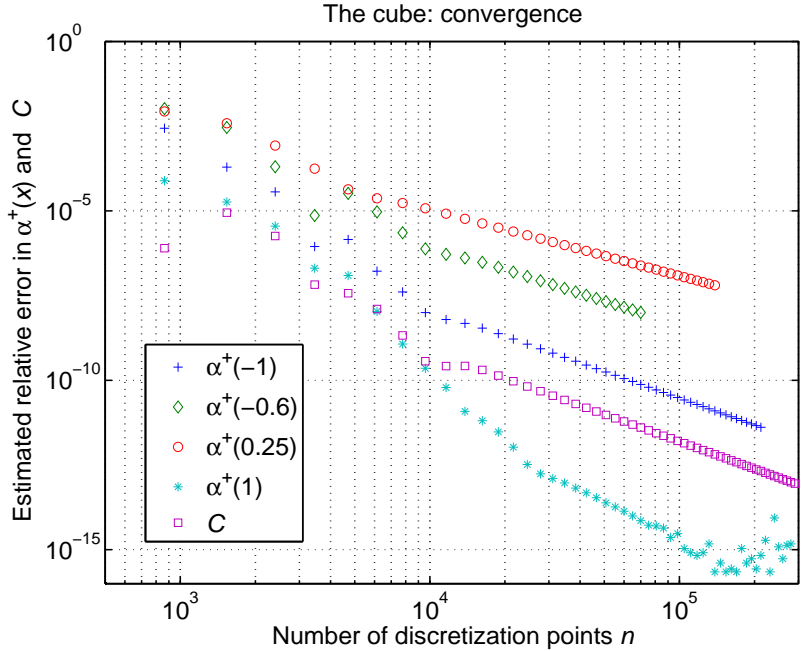}
 \caption{Convergence with the number of discretization points
   for the polarizability $\alpha^+(x)$ and capacitance $C$ of a unit
   cube. The values of $x$ correspond to: cube with infinite
   permittivity ($x=-1$), resonance in corners ($x=-0.6$), resonance
   along edges ($x=0.25$), and cube with zero permittivity ($x=1$).
   See Table~\ref{tab:refvals} for reference values used in error
   estimates.}
\label{fig:convergos}
\end{figure}

In our algorithm for $\alpha^+(x)$ via~(\ref{eq:final}), a complex
valued ${\bf R}$ generally implies a complex valued solution
$\tilde{\boldsymbol{\rho}}_{\rm coa}$, complex valued basis functions
${\bf Y}$, and $\Im\{\alpha^+(x)\}>0$. The square exhibits complex
valued limits of ${\bf R}_{\rm sq}$ whenever $z=x+{\rm
  i}y\in\mathbb{C}_+\to x\in\sigma_{{\mu}{\rm sq}}$. The cube, which
uses ${\bf R}_{\rm sq}$ and ${\bf Y}_{\rm sq}$ for the discretization
of $K$ and which in turn enters into the recursion~(\ref{eq:genrec})
for ${\bf R}_{\rm cu}$, will therefore have complex valued limits of
${\bf R}_{\rm cu}$ and $\Im\{\alpha^+(x)\}>0$ whenever
$x\in\sigma_{{\mu}{\rm sq}}$. We refer to this as {\it resonance along
  edges}. The reason being that $\Im\{\alpha^+(x)\}>0$ implies $x \in
\spec(K,H^{-1/2})$ and, as a consequence of the symmetrization
arguments in Section~\ref{sec:results} and the fact that the spectrum
of a self-adjoint operator consists of eigenvalues and approximate
eigenvalues, each such $x$ is either an eigenvalue or an approximate
eigenvalue of $K$ on $H^{-1/2}(S)$. Moreover, ${\bf R}_{\rm cu}$ may
remain complex valued throughout the numerical homotopy process as
$z=x+{\rm i}y\in\mathbb{C}_+\to x$ for other $x\in[-1,1]$ as well,
that is, where limits of ${\bf R}_{\rm sq}$ are real valued. We refer
to this as {\it resonance in corners}.

Fig.~\ref{fig:convergos} shows results from convergence tests. The
total number of discretization points on the cube surface is $n$. The
examples with $n\gtrsim 2\cdot 10^{4}$ were executed on a workstation
equipped with an IntelXeon E5430 CPU at 2.66 GHz and 32 GB of memory.
Out of the chosen values of $x$, one can see that the error in
$\alpha^+(x)$ is largest for $x=0.25$. This is not surprising. A
weakness in our algorithm is the assumption that the columns of ${\bf
  Y}_{\rm sq}$ are efficient basis functions for $\rho(r)$ in the
direction perpendicular to an edge. When $x=0.25$, there is resonance
along edges and the shortcomings of this assumption should be
particularly visible. When $x=-1$, $x=1$ and for $C$ there are real
valued solutions $\rho(r)\in L^2(S)$ which are easier to resolve.

\begin{table}[ht!]
\caption{Reference values, estimated relative errors, and best previous
  estimates for the limit polarizability $\alpha^+(x)$ and the
  capacitance $C$ of the cube.} 
\vspace{0.2cm}
\centering
\begin{tabular}{|l|l|l|l|c|r|} 
\hline  
   & present reference values & relerr  & previous results & Ref.
   & relerr \\
\hline 
$\alpha^+(-1)$   & $\;\;\: 3.644305190268$  & $10^{-11}$  &
                   $\;\;\: 3.6442$ & \cite{Sihvo04}   &  $3\cdot 10^{-5}$ \\
$\alpha^+(-0.6)$ & $\;\;\: 5.85574775+16.64205643{\rm i}$ & $10^{-8}$ &
                                                         & & \\
$\alpha^+(0.25)$ & $-2.76289925+3.08034035{\rm i}$  & $10^{-7}$ &
                                                         & & \\
$\alpha^+(1)$    & $-1.638415712936517$                & $10^{-14}$ &
                   $-1.6383$   & \cite{Sihvo04}  & $6\cdot 10^{-5}$ \\
$C$              & $\;\;\: 0.66067815409957$           & $10^{-13}$ &
                   $\;\;\: 0.66067813$ & \cite{Hwang10} & $10^{-7}$ \\
\hline
\end{tabular}
\label{tab:refvals} 
\end{table}

\begin{figure}[t]
  \centering 
\includegraphics[height=88mm]{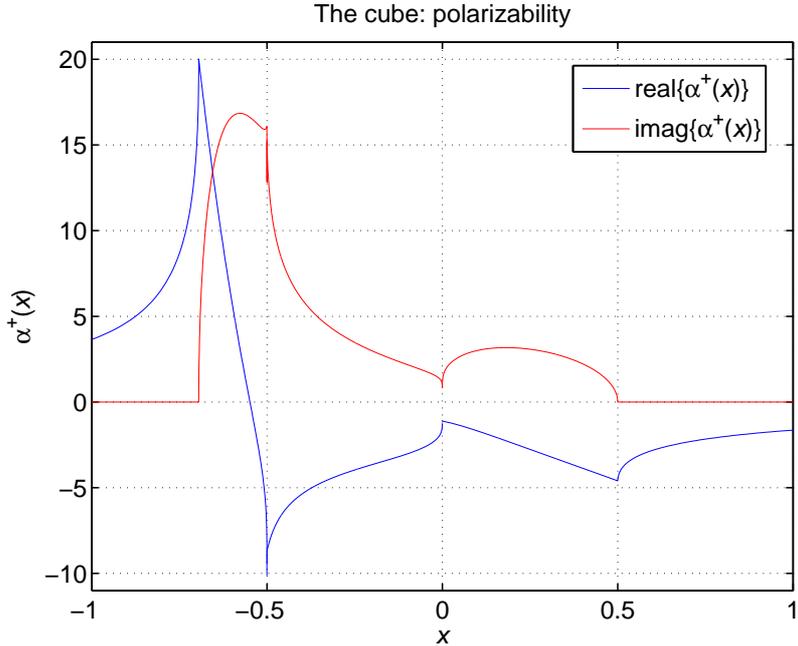}
 \caption{Polarizability of a cube. The curves are supported by
   1195 data points whose relative accuracies range from ten to five
   digits. No convergent results were obtained for $\alpha^+(x)$
   within a distance of $10^{-10}$ from $x=-0.5$ and within a distance
   of $10^{-5}$ from $x=0$. The maximum value of $\Re\{\alpha^+(x)\}$
   is $20.00826$ and occurs at $x\approx-0.694526$. The sum
   rules~(\ref{eq:sum1}), (\ref{eq:sum2}), and (\ref{eq:sum3}),
   evaluated via~(\ref{eq:herglotz}) using a composite trapezoidal
   rule, hold to a relative precision of $3\cdot 10^{-5}$, $4\cdot
   10^{-5}$, and $6\cdot 10^{-5}$, respectively.}
\label{fig:polabcu}
\end{figure}

The recursion~(\ref{eq:genrec}) for ${\bf R}_{\rm cu}$ requires a
substantial amount of memory when $n_{\rm p}$ is large and Newton's
method is activated. This explains why the test series for
$\alpha^+(-0.6)$ in Fig.~\ref{fig:convergos} had to be interrupted at
$n=69984$, that is, for $n_{\rm p}=27$. Timings vary greatly with $x$,
$n_{\rm p}$, and the error tolerances that are set in recursions and
iterative solvers. For safety, we set tolerances low and quote the
following approximate computing times for $n=9600$: $\alpha^+(-1)$,
$\alpha^+(1)$, and $C$ took one minute each; $\alpha^+(0.25)$ took two
minutes; $\alpha^+(-0.6)$ took ten minutes.  The test series for
$\alpha^+(-1)$, $\alpha^+(1)$ and $C$ all confirm previous available
results, see Table~\ref{tab:refvals}, and improve on these with
between six and nine digits.

Fig.~\ref{fig:polabcu} is our main numerical result. It shows that the
isolated cube has no bright plasmons (no poles in $\alpha^+(x)$). For
validation, see the figure caption, we have used the two sum rules
(\ref{eq:sum1}), (\ref{eq:sum2}) and a third sum rule
\begin{equation}
\int_\mathbb{R}x^2\mu'(x)\,{\rm d}x=0.433464767896\,,
\label{eq:sum3}
\end{equation}
which has been determined numerically for the cube~\cite{Hels91}.
Fig.~\ref{fig:polabcu} also shows that $\sigma_{{\mu}{\rm cu}}$ is
approximately equal to the interval $(-0.694526,0.5)$, possibly
punctured at $\{-0.5,0\}$, and raises the intriguing question of what
surface shapes $S$ lead to connected $\sigma_{\mu}$. Our computed
$\sigma_{{\mu}{\rm cu}}$ is broader than those estimated by earlier
investigators~\cite{Fuchs75,Klimo08,Lang76} using other techniques. We
conclude, again, that working with sharp edges and corners is an
advantage.

\section{Conclusions and outlook}

We have constructed an electrostatic solver for cube-shaped domains
and used it to produce new results for some canonical problems. A
particular characteristic of the solver is that it takes advantage of
sharp edges and corners, rather than being a victim of them.

A mathematical framework capturing the symmetric features of the
double layer potential and its adjoint has been identified, allowing
for analysis of the polarizability $\alpha(z)$ through spectral theory
for self-adjoint operators. Moreover, this framework and its
corresponding machinery is natural and satisfactory also from a
physical viewpoint, as the potentials produced to solve the
electrostatic problem are exactly those of finite energy. With full
mathematical rigor we have shown how to, with respect to
polarizabilities, interpret the limit process which occurs when
deforming a smooth surface into a cube, a point which has generated a
fair amount of discussion in the materials science community.

Furthermore, we have shown that while the polarizability via limits
can be extended to a function $\alpha^+(x)$ defined almost everywhere
on the real axis, the same statement fails for the corresponding
potentials whenever the representing measure $\mu$ has a non-zero
absolutely continuous part. That is, no finite energy potential can be
associated with a point $x \in \R$ where $\mu'(x) > 0$ exists and is
non-zero, neither as a direct solution to the electrostatic problem,
nor as a limit from the upper half-plane. We remark, however, that
while the density distributions $\rho_x$ hence fail to exist for such
$x$, the whole map $x\mapsto \rho_x$ can still be interpreted
naturally as a vector-valued distribution.

For the cube, the mathematical theory in conjunction with the
numerical findings show that ${\rm d}\mu(x) = \mu'(x) \, {\rm d}x$ is
purely absolutely continuous, and that the set where $\mu'(x)>0$ is
given by the interval $(a,b)\approx(-0.694526,0.5)$, possibly
excepting the two points $x=-0.5$ and $x=0$. Hence the integral
representation (\ref{eq:alpharep}) for $\alpha(z)$, in terms of
$\mu'(x)$, holds. Finally, $\mu'(x)$ has been determined numerically.
These discoveries, we hope, will be of use in plasmonics where
absorption peaks of nanocubes play an important role.

Future efforts will be directed towards solving the Helmholtz
equation, compare~\cite{Brem12a,Brem12b,Brem10b}. Should the
polarizability of clusters of cubes or the effective permittivity of
cubes in periodic arrangements be of interest, our solver requires
only minor modifications.

\section*{Acknowledgements}

We thank professors Alexandru Aleman, Anders Karlsson, Mihai Putinar, Daniel
Sj{\"o}berg, and Laurian Suciu for useful conversations. This work was
supported by the Swedish Research Council under contract 621-2011-5516.

\end{document}